\RequirePackage{rotating}
\documentclass{tlp}
\usepackage{epsfig}
\usepackage[utf8]{inputenc}
\usepackage{multicol}
\usepackage{multirow}
\usepackage{wrapfig}
\usepackage{amsmath}
\usepackage{rotating}
\usepackage{url}
\newtheorem{theorem}{Theorem}
\newtheorem{corollary}{Corollary}

\begin{document}

\label{firstpage}

\title[Table Space Designs For Implicit and Explicit Concurrent Tabled Evaluation]
      {Table Space Designs For Implicit and Explicit Concurrent Tabled Evaluation}

\author[Miguel Areias and Ricardo Rocha]
       {MIGUEL AREIAS and RICARDO ROCHA\\
       CRACS \& INESC-TEC and Faculty of Sciences, University of Porto\\
       Rua do Campo Alegre, 1021/1055, 4169-007 Porto, Portugal\\
       \email{\{miguel-areias,ricroc\}@dcc.fc.up.pt}}

\maketitle


\begin{abstract}
  One of the main advantages of Prolog is its potential for the
  \emph{implicit exploitation of parallelism} and, as a high-level
  language, Prolog is also often used as a means to \emph{explicitly
    control concurrent tasks}. Tabling is a powerful implementation
  technique that overcomes some limitations of traditional Prolog
  systems in dealing with recursion and redundant
  sub-computations. Given these advantages, the question that arises
  is if tabling has also the potential for the exploitation of
  concurrency/parallelism. On one hand, tabling still exploits a
  search space as traditional Prolog but, on the other hand, the
  concurrent model of tabling is necessarily far more complex since it
  also introduces concurrency on the access to the tables. In this
  paper, we summarize Yap's main contributions to concurrent tabled
  evaluation and we describe the design and implementation challenges
  of several alternative table space designs for implicit and explicit
  concurrent tabled evaluation which represent different trade-offs
  between concurrency and memory usage. We also motivate for the
  advantages of using \emph{fixed-size} and \emph{lock-free} data
  structures, elaborate on the key role that the engine's \emph{memory
    allocator} plays on such environments, and discuss how Yap’s
  mode-directed tabling support can be extended to concurrent
  evaluation. Finally, we present our future perspectives towards an
  efficient and novel concurrent framework which integrates both
  implicit and explicit concurrent tabled evaluation in a single
  Prolog engine. Under consideration in Theory and Practice of Logic
  Programming (TPLP).
\end{abstract}



\section{Introduction}

Tabling~\cite{Chen-96} is a recognized and powerful implementation
technique that overcomes some limitations of traditional Prolog
systems in dealing with recursion and redundant
sub-computations. Tabling is a refinement of SLD resolution that stems
from one simple idea: save intermediate answers from past computations
so that they can be reused when a \emph{similar call} appears during
the resolution process. Tabling based models are able to reduce the
search space, avoid looping, and always terminate for programs with
the bounded term-size property~\cite{Chen-96}.

Tabling has become a popular and successful technique thanks to the
ground-breaking work in the XSB Prolog system and in particular in the
SLG-WAM engine~\cite{Sagonas-98}, the most successful engine of
XSB. The success of SLG-WAM led to several alternative implementations
that differ in the execution rule, in the data-structures used to
implement tabling, and in the changes to the underlying Prolog
engine. Currently, the tabling technique is widely available in
systems like XSB Prolog~\cite{Swift-12}, Yap Prolog~\cite{CostaVS-12},
B-Prolog~\cite{Zhou-12}, ALS Prolog~\cite{Guo-01},
Mercury~\cite{Somogyi-06}, Ciao Prolog~\cite{Chico-08} and more
recently in SWI Prolog~\cite{Desouter-15} and Picat~\cite{Zhou-15}.

One of the main advantages of Prolog is its potential for the
\emph{implicit exploitation of parallelism}. Many sophisticated and
well-engineered parallel Prolog systems exist in the
literature~\cite{Gupta-01}, being the most successful those that
exploit \emph{implicit
  or-parallelism}~\cite{Aurora-88,Ali-90a,Gupta-99}, \emph{implicit
  and-parallelism}~\cite{Hermenegildo-91,Shen-92,Pontelli-97} or a
combination of both~\cite{CostaVS-91}. Or-parallelism arises when more
than one clause unifies with the current call and it corresponds to
the simultaneous execution of the body of those different
clauses. And-parallelism arises when more than one subgoal occurs in
the body of the clause and it corresponds to the simultaneous
execution of the subgoals contained in a clause's body.

On the other hand, as a high-level language, Prolog is also often used
as a means to \emph{explicitly control and schedule concurrent
  tasks}~\cite{Carro-99,Fonseca-09a}. The ISO Prolog multithreading
standardization proposal~\cite{Moura-08b} is currently implemented in
several Prolog systems including XSB, Yap, Ciao and SWI, providing a
highly portable solution given the number of operating systems
supported by these systems. In a nutshell, multithreading in Prolog is
the ability to concurrently perform multiple computations, in which
each computation runs independently but shares the database
(clauses). It is therefore unsurprising that implicit and explicit
concurrent/parallel evaluation has been an important subject in the
design and development of Prolog systems.

Nowadays, the increasing availability of computing systems with
multiple cores sharing the main memory is already a standardized,
high-performance and viable alternative to the traditional (and often
expensive) shared memory architectures. The number of cores per
processor is expected to continue to increase, further expanding the
potential for taking advantage of such support as an increasingly
popular way to implement dynamic, highly asynchronous, concurrent and
parallel programs. 

Besides the two traditional approaches to concurrency/parallelism: (i)
\emph{fully implicit}, i.e., it is left to the runtime system to
automatically detect the potential concurrent tasks in the program,
assign them for parallel execution and control and synchronize their
execution; and (ii) \emph{fully explicit}, i.e., it is left to the
user to annotate the tasks for concurrent execution, assign them to
the available workers and control the execution and the
synchronization points, the recent years have seen a lot of proposals
trying to combine both approaches in such a way that the user relies
on high-level explicit parallel constructs to trigger parallel
execution and then it is left to the runtime system the control of the
low-level execution details. In the combined approach, in general, a
program begins as a single worker that executes sequentially until
reaching a \emph{parallel construct}. When reaching a parallel
construct, the runtime system launches a set of additional workers to
exploit concurrently the sub-computation at hand. Concurrent execution
is then handled implicitly by the execution model taking into account
additional directive restrictions given to the parallel construct.

Multiple examples of frameworks exist that follow the combined
approach. For example, for imperative programming languages, the
OpenMP~\cite{Chapman-08}, Intel Threading Building
Blocks~\cite{Reinders-07} and Cilk~\cite{Blumofe-1995} frameworks
provide runtime systems for multithreaded parallel programming,
providing users with the means to create, synchronize, and schedule
threads efficiently. For functional programming languages, the
Eden~\cite{Loogen-05} and HDC~\cite{Herrmann-00} Haskell based
frameworks allow the users to express their programs using polymorphic
higher-order functions. For object-oriented programming languages,
MALLBA~\cite{Alba-02} and DPSKEL~\cite{Pelaez-07} frameworks also
showed relevant speedups in the parallel evaluation of combinatorial
optimization benchmarks.

In the specific case of Prolog, given the advantages of tabled
evaluation, the question that arises is if a tabling mechanism has the
potential for the exploitation of concurrency/parallelism. On one
hand, tabling still exploits a search space as traditional Prolog, but
on the other hand, the concurrent model of tabling is necessarily far
more complex than the traditional concurrent models, since it also
introduces concurrency on the access to the tables. In a concurrent
tabling system, tables may be either \emph{private} or \emph{shared}
between workers. On one hand, private tables can be easier to
implement but restrict the degree of concurrency. On the other hand,
shared tables have all the associated issues of locking,
synchronization and potential deadlocks. Here, the problem is even
more complex because we need to ensure the correctness and
completeness of the answers found and stored in the shared
tables. Thus, despite the availability of both threads and tabling in
Prolog compilers such as XSB, Yap, Ciao and SWI, the implementation of
these two features such that they work together seamlessly implies
complex ties to one another and to the underlying engine.

To the best of our knowledge, only the XSB and Yap systems support the
combination of tabling with some form of concurrency/parallelism. In
XSB, the SLG-WAM execution model was extended with a \emph{shared
  tables design}~\cite{Marques-08} to support explicit concurrent
tabled evaluation using threads. It uses a semi-naive approach that,
when a set of subgoals computed by different threads is mutually
dependent, then a \emph{usurpation operation} synchronizes threads and
a single thread assumes the computation of all subgoals, turning the
remaining threads into consumer threads. The design ensures the
correct execution of concurrent sub-computations but the experimental
results showed some limitations~\cite{Marques-10}. Yap implements both
implicit and explicit concurrent tabled evaluation, but
separately. The OPTYap design~\cite{Rocha-05a} combines the
tabling-based SLG-WAM execution model with implicit or-parallelism
using shared memory processes. More recently, a second design supports
explicit concurrent tabled evaluation using threads~\cite{Areias-12a},
but using an alternative view to XSB's design. In Yap's design, each
thread has its own tables, i.e., from a thread point of view the
tables are private, but at the engine level it uses a \emph{common
  table space}, i.e., from the implementation point of view the tables
are shared among threads.

In this paper, we summarize Yap's main developments and contributions
to concurrent tabled evaluation and we describe the design and
implementation challenges of several alternative table space designs
for implicit and explicit concurrent tabled evaluation which represent
different trade-offs between concurrency and memory usage. We also
motivate for the advantages of using \emph{fixed-size} and
\emph{lock-free} data structures for concurrent tabling and we
elaborate on the key role that the engine's \emph{memory allocator}
plays on such an environment where a higher number of simultaneous
memory requests for data structures in the table space can be made by
multiple workers. We also discuss how Yap’s mode-directed tabling
support~\cite{Santos-13} can be extended to concurrent
evaluation. Mode-directed tabling is an extension to the tabling
technique that allows the aggregation of answers by specifying
pre-defined modes such as \emph{min} or \emph{max}. Mode-directed
tabling can be viewed as a natural tool to implement dynamic
programming problems, where a general recursive strategy divides a
problem in simple sub-problems whose goal is, usually, to dynamically
calculate optimal or selective answers as new results arrive.

Finally, we present our future perspectives towards an efficient and
novel concurrent framework which integrates both implicit and explicit
concurrent tabled evaluations in a single tabling engine. This is a
very complex task since we need to combine the explicit control
required to launch, assign and schedule tasks to workers, with the
built-in mechanisms for handling tabling and/or implicit concurrency,
which cannot be controlled by the user. Such a framework could renew
the glamour of Prolog systems, especially in the concurrent/parallel
programming community. Combining the inherent implicit parallelism of
Prolog with explicit high-level parallel constructs will clearly
enhance the expressiveness and the declarative style of tabling, and
simplify concurrent programming.


In summary, the main contributions of this paper are: (i) a systematic
presentation of the different alternative table space designs
implemented in Yap for implicit and explicit concurrent tabled
evaluation (which were dispersed by several publications); (ii) a
formalization of the total memory usage of each table space design,
which allows for a more rigorous comparison and demonstrates how each
design is dependent on the number of workers and on the number of
tabled calls in evaluation; (iii) a performance analysis of Yap's
tabling engine highlighting how independent concurrent flows of
execution interfere at the low-level engine and how dynamic
programming problems fit well with concurrent tabled evaluation; and
(iv) the authors' perspectives towards a future concurrent framework
which integrates both implicit and explicit concurrent tabled
evaluations in a single tabling engine.


The remainder of the paper is organized as follows. First, we
introduce some basic concepts and relevant background. Then, we
present the alternative table space designs for implicit and explicit
concurrent tabled evaluation. Next, we discuss the most important
engine components and implementation challenges to support concurrent
tabled evaluation and we show a performance analysis of Yap's tabling
engine when using different table space designs. At last, we discuss
future perspectives and challenging research directions.


\section{Background}

This section introduces relevant background needed for the following
sections. It briefly describes Yap's original table space organization
and presents Yap's approach for supporting mode-directed tabling.


\subsection{Table Space Organization}

The basic idea behind tabling is straightforward: programs are
evaluated by saving intermediate answers for tabled subgoals so that
they can be reused when a \emph{similar call} appears during the
resolution process. First calls to tabled subgoals are considered
\emph{generators} and are evaluated as usual, using SLD resolution,
but their answers are stored in a global data space, called the
\emph{table space}. Similar calls are called \emph{consumers} and are
resolved by consuming the answers already stored for the corresponding
generator, instead of re-evaluating them against the program
clauses. During this process, as further new answers are found, they
are stored in their table entries and later returned to all similar
calls.

Call similarity thus determines if a subgoal will produce their own
answers or if it will consume answers from a generator call. There are
two main approaches to determine if a subgoal $A$ is similar to a
subgoal $B$:

\begin{itemize}
\item \emph{Variant-based tabling}~\cite{RamakrishnanIV-99}: $A$ and $B$
  are variants if they can be identical through variable renaming. For
  example, $p(X,1,Y)$ and $p(W,1,Z)$ are \emph{variants} because both
  can be renamed into $p(VAR_0,1,VAR_1)$.
\item \emph{Subsumption-based tabling}~\cite{Rao-96}: subgoal $A$ is
  considered similar to $B$ if $A$ is \emph{subsumed} by $B$ (or $B$
  \emph{subsumes} $A$), i.e., if $A$ is more specific than $B$ (or an
  instance of).  For example, subgoal $p(X,1,2)$ is subsumed by
  subgoal $p(Y,1,Z)$ because there is a substitution $\{Y=X,Z=2\}$
  that makes $p(X,1,2)$ an instance of $p(Y,1,Z)$.
\end{itemize}
  
Variant-based tabling has been researched first and is arguably better
understood. For some types of programs, subsumption-based tabling
yields superior time performance~\cite{Rao-96,Johnson-99}, as it
allows greater reuse of answers, and better space usage, since the
answer sets for the subsumed subgoals are not stored. However, the
mechanisms to efficiently support subsumption-based tabling are harder
to implement, which makes subsumption-based tabling not as popular as
variant-based tabling. The Yap Prolog system implements both
approaches for sequential tabling~\cite{CostaVS-12,Cruz-10}, but for
concurrent tabled evaluation, Yap follows the variant-based tabling
approach.

\begin{wrapfigure}{R}{7.5cm}
\includegraphics[width=7.5cm]{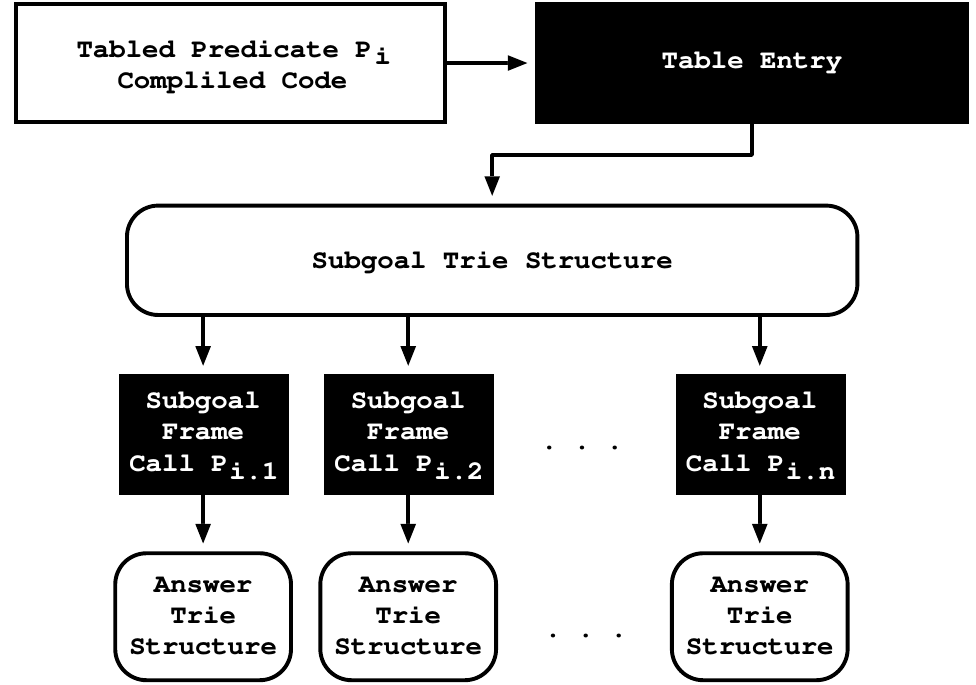}
\caption{Yap's original table space organization}
\label{fig_table_space_original}
\vspace{-\intextsep}
\end{wrapfigure}

A critical component in the implementation of an efficient tabling
system is thus the design of the data structures and algorithms to
access and manipulate the table space. Yap uses \emph{trie data
structures} to implement efficiently the table
space~\cite{RamakrishnanIV-99}. Tries are trees in which common
prefixes are represented only once. The trie data structure provides
complete discrimination for terms and permits lookup and possible
insertion to be performed in a single pass through a term, hence
resulting in a very efficient and compact data structure for term
representation.

Figure~\ref{fig_table_space_original} shows the original table space
organization for a tabled predicate $P_i$ in Yap. At the entry point,
we have the \emph{table entry} data structure. This structure stores
common information for the tabled predicate, such as the predicate's
arity or the predicate's evaluation strategy, and it is allocated when
the predicate is being compiled, so that a pointer to the table entry
can be included in its compiled code. This guarantees that further
calls to the predicate will access the table space starting from the
same point. Below the table entry, we have the \emph{subgoal trie
  structure}. Each different tabled subgoal call $P_{i.j}$ to the
predicate corresponds to a unique path through the subgoal trie
structure, always starting from the table entry, passing by several
subgoal trie data units, the \emph{subgoal trie nodes}, and reaching a
leaf data structure, the \emph{subgoal frame}. The subgoal frame
stores additional information about the subgoal and acts like an entry
point to the \emph{answer trie structure}. Each unique path through
the answer trie data units, the \emph{answer trie nodes}, corresponds
to a different tabled answer to the entry subgoal.


\subsection{Mode-Directed Tabling and Dynamic Programming}

The tabling technique can be viewed as a natural tool to implement
dynamic programming problems. Dynamic programming is a general
recursive strategy that consists in dividing a problem in simple
sub-problems that, often, are the same. Tabling is thus suitable to
use with this kind of problems since, by storing and reusing
intermediate results while the program is executing, it avoids
performing the same computation several times.

In a traditional tabling system, all arguments of a tabled subgoal
call are considered when storing answers into the table space. When a
new answer is not a variant of any answer that is already in the table
space, then it is always considered for insertion. Therefore,
traditional tabling is very good for problems that require storing all
answers. However, with dynamic programming, usually, the goal is to
dynamically calculate optimal or selective answers as new results
arrive. Solving dynamic programming problems can thus be a difficult
task without further support.

\emph{Mode-directed tabling} is an extension to the tabling technique
that supports the definition of \emph{modes} for specifying how
answers are inserted into the table space. Within mode-directed
tabling, tabled predicates are declared using statements of the form
`$table~p(m_1,...,m_n)$', where the $m_i$’s are \emph{mode operators}
for the arguments. The idea is to define the arguments to be
considered for variant checking (the index arguments) and how variant
answers should be tabled regarding the remaining arguments (the output
arguments). In Yap, index arguments are represented with mode
\emph{index}, while arguments with modes \emph{first}, \emph{last},
\emph{min}, \emph{max}, \emph{sum} and \emph{all} represent output
arguments~\cite{Santos-13}. After an answer is generated, the system
tables the answer only if it is \emph{preferable}, accordingly to the
meaning of the output arguments, than some existing variant answer.

In Yap, mode-directed tabled predicates are compiled by extending the
table entry data structure to include a \emph{mode array}, where the
information about the modes is stored, and by extending the subgoal
frames to include a \emph{substitution array}, where the mode
information is stored together with the number of free variables
associated with each argument in the subgoal
call~\cite{Santos-13}. When a new answer is found, it must be compared
against the answer(s) already stored in the table, accordingly to the
modes defined for the corresponding arguments. If the new answer is
preferable, the old answer(s) must be \emph{invalidated} and the new
one inserted in the table. The invalidation process consists in: (a)
deleting all intermediate answer trie nodes corresponding to the
answers being invalidated; and (b) tagging the leaf nodes of such
answers as invalid nodes. Invalid nodes are only deleted when the
table is later completed or abolished.

Regarding the table space designs that we present next, the support
for mode-directed tabling is straightforward when the table data
structures are not accessed concurrently for write operations. The
problem arises for the designs which do not require the completion of
tables to share answers, since we need to efficiently support
concurrent delete operations on the trie structures and correctly
handle the interface between consumer calls and the navigation in the
answer tries.


\section{Concurrent Table Space Designs}
\label{sec_concurrent_table_designs}
  
This section presents alternative table space designs for implicit and
explicit concurrent tabled evaluation, which represent different
trade-offs between concurrency and memory usage.


\subsection{Implicit versus Explicit Concurrent Tabled Evaluation}
\label{sec_implicit_explicit}
  
Remember the two traditional approaches to concurrency/parallelism:
\emph{fully implicit} and \emph{fully explicit}. With fully implicit,
it is left to the runtime system to automatically detect the potential
concurrent tasks in the program, assign them for concurrent/parallel
execution and control and synchronize their execution. In such
approach, the running workers (processes, threads or both) often share
the data structures representing the data of the problem since tasks
do not need to be pre-assigned to workers as any worker can be
scheduled to perform an unexplored concurrent task of the problem. For
tabling, that means that the table space data structures must be fully
shared among all workers. This is the case of the OPTYap
design~\cite{Rocha-05a}, which combines the tabling-based SLG-WAM
execution model with implicit or-parallelism using shared memory
processes.

On the other hand, with a fully explicit approach, it is left to the
user to annotate the tasks for concurrent execution, assign them to
the available workers and control the execution and the synchronization
points. In such approach, the running workers often execute
independently a well-defined (set of) task(s). For tabling, that means
that each evaluation only depends on the computations being performed
by the worker itself, i.e., a worker does not need to consume answers
from other workers' tables as it can always be the generator for all
of its subgoal calls. These are the cases of XSB~\cite{Marques-08} and
Yap~\cite{Areias-12a} designs which support explicit concurrent tabled
evaluation using threads. In any case, the table space data structures
can be either private or partially shared between workers. Yap
proposes several alternative designs to implement the table space for
explicit concurrent tabled
evaluation. Table~\ref{tab_table_space_summary} overviews the several
Yap's table space designs and how they differ in the way the internal
table data structures are implemented and accessed. In the following
subsections, we present the several designs and we show a detailed
analysis of the memory usage of each.

\begin{table}[!ht]
\centering
\caption{Yap's table space designs -- Cooperative Sharing (CS),
  No-Sharing (NS), Subgoal-Sharing (SS), Full-Sharing (FS), Partial
  Answer Sharing (PAS) and Private Answer Chaining (PAC) -- and the
  implementation and access of the data structures in each design: as
  private data structures (--); as fully shared data structures (F);
  as partially shared data structures (P); and as data structures with
  concurrent read (r) and concurrent write (w) operations.}
\begin{tabular}{c|cccccc}
\hline
\textbf{\emph{Data}} & \textbf{\emph{Implicit}} & \multicolumn{4}{c}{\textbf{\emph{Explicit}}} \\ \cline{3-7}
\textbf{\emph{Structure}} & \textbf{\emph{CS}} & \textbf{\emph{NS}} & \textbf{\emph{SS}} & \textbf{\emph{FS}} & \textbf{\emph{PAS}} & \textbf{\emph{PAC}} \\
\hline\hline
\textbf{\emph{Table Entry}}   & $F(r)$  & $F(r)$ & $F(r)$  & $F(r)$  & $F(r)$  & $F(r)$  \\
\textbf{\emph{Subgoal Trie}}  & $F(rw)$ & --     & $F(rw)$ & $F(rw)$ & $F(rw)$ & $F(rw)$ \\
\textbf{\emph{Subgoal Frame}} & $F(rw)$ & --     & --      & $P(rw)$ & $P(r)$  & $P(rw)$ \\
\textbf{\emph{Answer Trie}}   & $F(rw)$ & --     & --      & $F(rw)$ & $P(r)$  & $P(rw)$ \\
\hline
\end{tabular}
\label{tab_table_space_summary}
\end{table}


\subsection{Cooperative Sharing Design}

The \emph{Cooperative Sharing (CS)} design supports the combination of
tabling with implicit or-parallelism using shared memory
processes~\cite{Rocha-05a}. The CS design was the first concurrent
table space design implemented in Yap Prolog. It follows Yap's
original table space organization, as shown in
Fig.~\ref{fig_table_space_original}, and extends it with some sort of
synchronization mechanisms to deal with concurrent accesses. In what
follows, we will not consider synchronization mechanisms which require
extending the table space data structures with extra fields, like lock
fields, since several synchronization techniques exist that do not
require an actual lock field. Two examples are: (i) the usage of an
external global array of locks; or (ii) the usage of low level
\emph{Compare-And-Swap (CAS)} operations. We discuss this in more
detail in section~\ref{sec_engine_components}.

Remember from Fig.~\ref{fig_table_space_original} that, at the entry
point, we have a table entry ($TE$) data structure for each tabled
predicate $P_i$. Underneath each $TE$, we have a subgoal trie
($ST(P_i)$) and several subgoal frame ($SF$) data structures for each
tabled subgoal call $P_{i.j}$ made to the predicate. Finally,
underneath each $SF$, we have an answer trie ($AT(P_{i.j})$) structure
with the answers for the corresponding subgoal call $P_{i.j}$. Please
note that the size of the $TE$ and $SF$ data structures is fixed and
independent from the predicate, but the size of the $ST(P_i)$ and
$AT(P_{i.j})$ data structures varies accordingly to the number of
subgoal calls made and answers found during tabled evaluation.

We can now formalize the \emph{Total Memory Usage (TMU)} of the CS
design. For this, we assume that all tabled predicates are completely
evaluated, meaning that the engine will not allocate any further data
structures on the table space. Given $NP$ tabled predicates,
Eq.~\ref{equation_tmu_cs} presents the $TMU$ of the CS design
($TMU_{CS}$).

\begin{equation}
\begin{aligned}
& TMU_{CS} = \sum\limits_{i = 1}^{NP} MU_{CS}(P_i) \\
& where~~
MU_{CS}(P_i) = TE + ST(P_i) + \sum\limits^{NC(P_i)}_{j=1} [SF +
    AT(P_{i.j})]
\end{aligned}
\label{equation_tmu_cs}
\end{equation} 

The $TMU_{CS}$ is given by the summation of the \emph{Memory Usage
  (MU)} of each predicate $P_i$, i.e, the $MU_{CS}(P_i)$ values, which
correspond then to the sum of each structure inside the table space
for the corresponding predicate $P_i$. The $TE$, $ST(P_i)$, $SF$ and
$AT(P_{i.j})$ values represent the amount of the memory used by
predicate $P_i$ in its table entry, subgoal trie, subgoal frames and
answer trie structures, respectively, and the $NC(P_i)$ value
represents the number of diferent tabled subgoal calls made to the
predicate. For example, in Fig.~\ref{fig_table_space_original}, the
value of $NC(P_i)$ is $n$.

As a final remark, please note that the total memory usage of the CS
design ($TMU_{CS}$) is the same as the total memory usage of Yap's
original table space organization ($TMU_{ORIG}$). Thus, in what
follows, we will use the $TE$, $ST(P_i)$, $SF$ and $AT(P_{i.j})$
values as the reference values for comparison against the other
concurrent table space designs.


\subsection{No-Sharing Design}

Yap implements explicit concurrent tabled evaluation using threads in
which each thread's computation only depends on the evaluations being
performed by the thread itself. The \emph{No-Sharing (NS)} design was
the starting design for supporting explicit concurrent tabled
evaluation in Yap~\cite{Areias-12a}. In the NS design, each thread
allocates fully private tables for each new tabled subgoal being
called. In this design, only the $TE$ structure is shared among
threads. Figure~\ref{fig_table_space_no_sharing} shows the
configuration of the table space for the NS design. For the sake of
simplicity, the figure only shows the configuration for a particular
predicate $P_i$ and a particular subgoal call $P_{i.j}$.

\begin{wrapfigure}{R}{7.5cm}
\includegraphics[width=7.5cm]{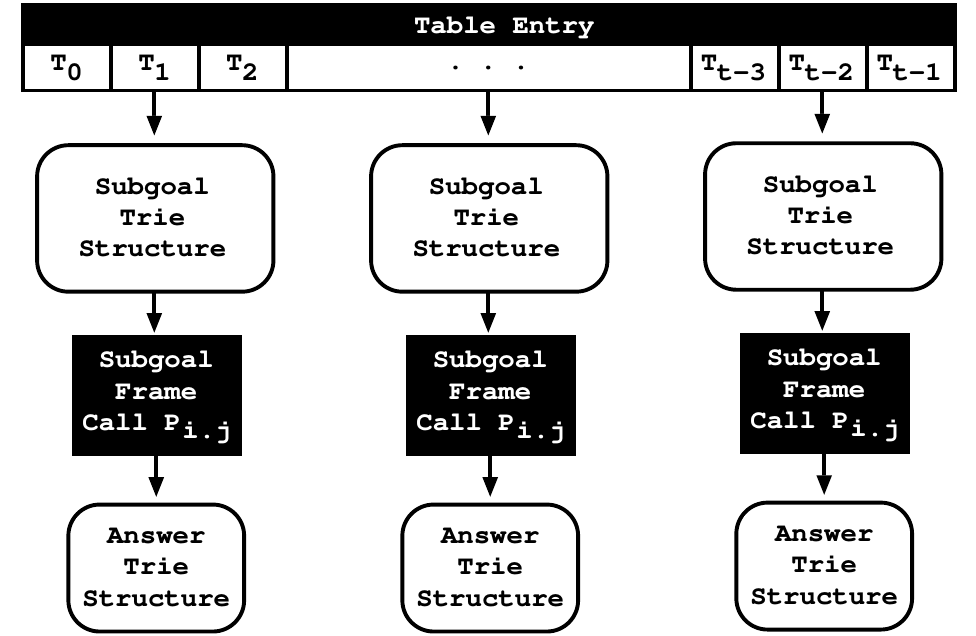}
\caption{Table space organization for the NS design}
\label{fig_table_space_no_sharing}
\vspace{-\intextsep}
\end{wrapfigure}


The table entry still stores the common information for the predicate
but it is extended with a bucket array ($BA$), where each thread $T_k$
has its own entry, which then points to the private $ST(P_i)$, $SF$
and $AT(P_{i.j})$ data structures of the thread. Each bucket array
contains as much entry cells as the maximum number of threads that can
be created in Yap (currently, Yap supports 1024 simultaneous
threads). However, in practice, this solution can be highly
inefficient and memory consuming, as this huge bucket array must be
always allocated even when only one thread will use it. To solve this
problem, we introduce a kind of \emph{inode pointer structure}, where
the bucket array is split into direct bucket cells and indirect bucket
cells~\cite{Areias-12a}. The direct bucket cells are used as before,
but the indirect bucket cells are allocated only as needed, which
alleviates the memory problem and easily adjusts to a higher maximum
number of threads. This direct/indirect organization is applied to all
bucket arrays.


Since the $ST(P_i)$, $SF$ and $AT(P_{i.j})$ data structures are
private to each thread, they can be removed when the thread finishes
execution. Only the table entry is shared among threads. As this
structure is created by the main thread when a program is being
compiled, no concurrent writing operations will exist between threads
and thus no synchronization points are required for the NS design.

Given an arbitrary number of $NT$ running threads and assuming that
all threads have completely evaluated the same number $NC(P_i)$ of
tabled subgoal calls, Eq.~\ref{equation_mu_ns} shows the memory usage
for a predicate $P_i$ in the NS design ($MU_{NS}(P_i)$).

\begin{equation}
\begin{aligned}
& MU_{NS}(P_i) = TE_{NS} + NT * [ST(P_i) + \sum\limits^{NC(P_i)}_{j=1} [SF + AT(P_{i.j})]] \\ 
& where~~
TE_{NS}  = TE + BA
\end{aligned}
\label{equation_mu_ns}
\end{equation} 

The $MU_{NS}(P_i)$ value is given by the sum of the memory size of the
extended table entry structure ($TE_{NS}$) plus the sum of the sizes
of the private structures of each thread multiplied by the $NT$
threads. The memory size of $TE_{NS}$ is given by the size of the
original $TE$ structure added with the memory size of the bucket array
($BA$). The memory size of the remaining structures is the same as in
Yap's original table space organization.

As for Eq.~\ref{equation_tmu_cs}, the total memory usage of the NS
design ($TMU_{NS}$) (not shown in Eq.~\ref{equation_mu_ns}) is given
by the summation of the memory usage of each predicate, i.e, the
$MU_{NS}(P_i)$ values. Comparing $TMU_{NS}$ with $TMU_{ORIG}$ given
$NP$ tabled predicates, the extra memory cost of the NS design to
support concurrency is given by the formula:

\begin{equation*}
\sum^{NP}_{i = 1}[BA + [NT - 1] * [ST(P_i) + \sum\limits^{NC(P_i)}_{j=1}
    [SF + AT(P_{i.j})]]]
\end{equation*} 

The formula shows that for the base case of 1 thread ($NT=1$), the
amount of extra memory spent by the NS design, given by $NP*BA$,
corresponds to the bucket array extensions. When increasing the number
of threads, the amount of extra memory spent in the $ST(P_i)$, $SF$
and $AT(P_{i.j})$ data structures increases proportionally to
$NT$. This dependency on the number of threads motivated us to create
alternative designs that could decrease the amount of extra memory to
be spent. The following subsections present such alternative designs.


\subsection{Subgoal-Sharing Design}

In the \emph{Subgoal-Sharing (SS)} design, the threads share part of
the table space. Figure~\ref{fig_table_space_subgoal_sharing} shows
the configuration of the table space for the SS design. Again, for the
sake of simplicity, the figure only shows the configuration for a
particular tabled predicate $P_i$ and a particular subgoal call
$P_{i.j}$.

\begin{wrapfigure}{R}{7.5cm}
\includegraphics[width=7.5cm]{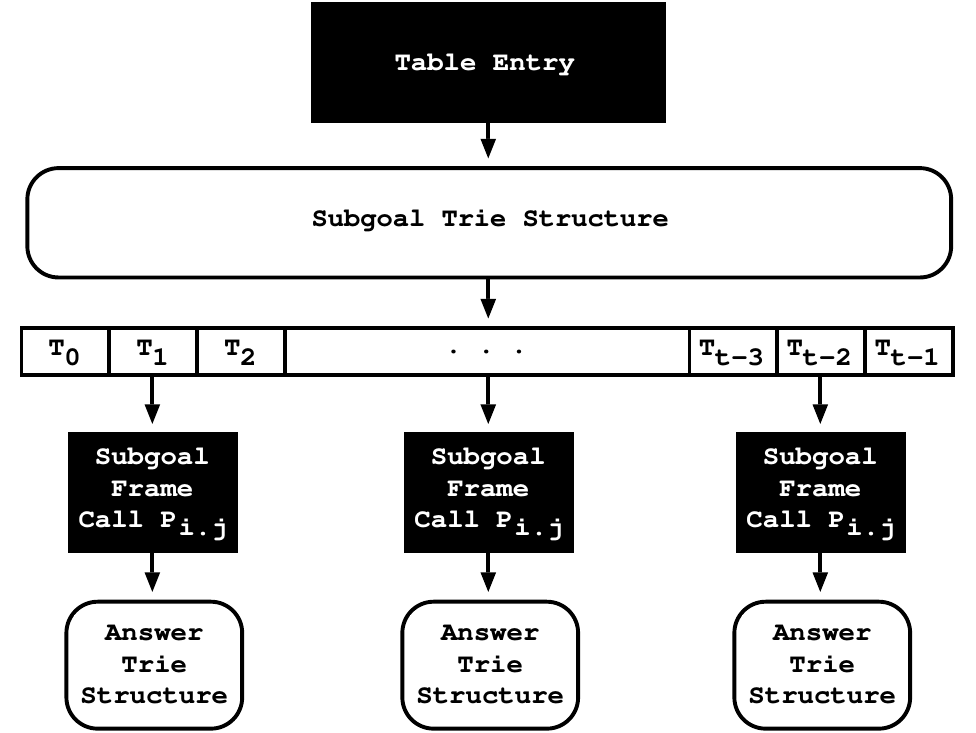}
\caption{Table space organization for the SS design}
\label{fig_table_space_subgoal_sharing}
\end{wrapfigure}

In the SS design, the $ST(P_i)$ data structures are now shared among
the threads and the leaf data structure in each subgoal trie path,
instead of referring a $SF$ as before, it now points to a $BA$. Each
thread $T_K$ has its own entry inside the $BA$ which then points to
private $SF$ and $AT(P_{i.j})$ structures. In this design, concurrency
among threads is restricted to the allocation of trie nodes on the
$ST(P_i)$ structures. Whenever a thread finishes execution, its
private structures are removed, but the shared part remains present as
it can be in use or be further used by other threads.

Given an arbitrary number of $NT$ running threads and assuming that
all threads have completely evaluated the same number $NC(P_i)$ of
tabled subgoal calls, Eq.~\ref{equation_mu_ss} shows the memory usage
for a predicate $P_i$ in the SS design ($MU_{SS}(P_i)$).

The memory usage for the SS design is given by the sum of the memory
size of the $TE$ and $ST(P_i)$ data structures plus the summation, for
each subgoal call, of the memory used by the $BA$ added with the sizes
of the private structures of each thread multiplied by the $NT$
threads. The memory size of each particular data structure is the same
as in Yap’s original table space organization.

\begin{equation}
\begin{aligned}
& MU_{SS}(P_i) = TE + ST(P_i) + \sum\limits^{NC(P_i)}_{j=1} [BA + NT * [SF + AT(P_{i.j})]]
\end{aligned}
\label{equation_mu_ss}
\end{equation} 

Theorem~\ref{theorem_NS_vs_SS} shows the conditions where the SS design
uses less memory than the NS design for an arbitrary number of threads
$NT$ and an arbitrary number of subgoal calls $NC(P_i)$\footnote{The
  proofs for all the theorems that follow are presented in detail
  in~\ref{appendix_proofs}.}.

\begin{theorem}
\label{theorem_NS_vs_SS}
If $NT \geq 1$ and $NC(P_i) \geq 1$ then $MU_{SS}(P_i) \leq
MU_{NS}(P_i) $ if and only if the formula $ [NC(P_i) - 1] * BA \leq
[NT - 1] * ST(P_i)$ holds.
\end{theorem}

Theorem~\ref{theorem_NS_vs_SS} shows that the comparison between the
NS and SS designs depends directly on the number of subgoal calls
($NC(P_i)$) made to the predicate by the number of threads ($NT$) in
evaluation. These numbers will affect the memory size of the $BA$ and
$ST(P_i)$ structures. The NS design grows in the number of $ST(P_i)$
structures as we increase the number of threads. The SS design grows
in the number of $BA$ structures proportionally to the number of
subgoal calls made to the predicate. The number of subgoal calls and
the size of the $ST(P_i)$ structures depends on the predicate being
evaluated, while the size of the $BA$ structures is fixed by the
implementation and the number of threads is user-dependent. For one
thread ($NT=1$), the following corollaries can be derived from
Thm.~\ref{theorem_NS_vs_SS}:

\begin{corollary}
If $NT = 1$ and $NC(P_i) = 1$ then $MU_{SS}(P_i) = MU_{NS}(P_i)$.
\end{corollary}

\begin{corollary}
If $NT = 1$ and $NC(P_i)>1$ then $MU_{SS}(P_i) > MU_{NS}(P_i)$.
\end{corollary}

In summary, for one thread, the SS design is equal to or worse than
the NS design in terms of memory usage. For a number of threads higher
than one, the SS design performs better than the NS design when the
formula in Thm.~\ref{theorem_NS_vs_SS} holds. The best scenarios for
the SS design occur for predicates with few subgoal calls and for
subgoal trie structures using larger amounts of memory. In such
scenarios, the difference between both designs increases
proportionally to the number of threads.


\subsection{Full-Sharing Design}

The \emph{Full-Sharing (FS)} design tries to maximize the amount of
data structures being shared among
threads. Figure~\ref{fig_table_space_full_sharing} shows the
configuration of the table space for the FS design. Again, for the
sake of simplicity, the figure only shows the configuration for a
particular tabled predicate $P_i$ and a particular subgoal call
$P_{i.j}$.

\begin{wrapfigure}{R}{8.5cm}
\includegraphics[width=8.5cm]{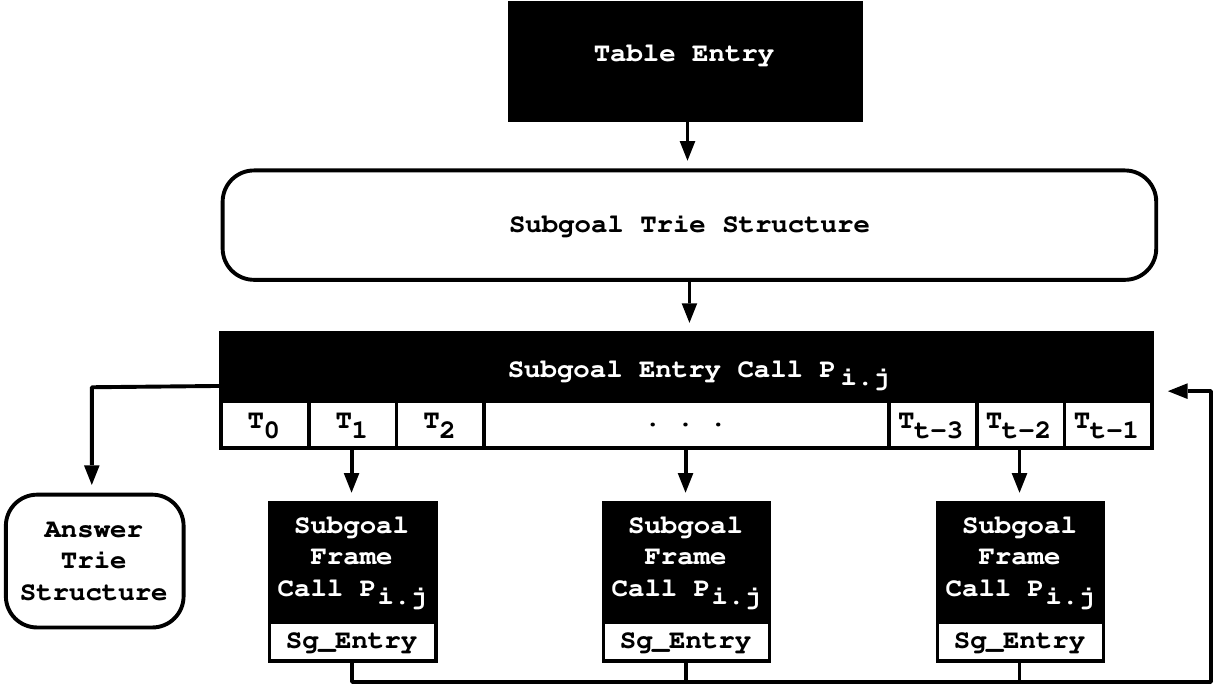}
\caption{Table space organization for the FS design}
\label{fig_table_space_full_sharing}
\vspace{-\intextsep}
\end{wrapfigure}

In this design, the $AT(P_{i.j})$ structure and part of the subgoal
frame information, the subgoal entry data structure in
Fig.~\ref{fig_table_space_full_sharing}, are now also shared among all
threads. The previous $SF$ data structure was split into two: the
subgoal entry stores common information for the subgoal call (such as
the pointer to the shared $AT(P_{i.j})$ structure) and a $BA$
structure; and the remaining information (the subgoal frame data
structure in Fig.~\ref{fig_table_space_full_sharing}) is kept private
to each thread. Concurrency among threads now includes also the access
to the subgoal entry data structure and the allocation of trie nodes
on the $AT(P_{i.j})$ structures.

The subgoal entry includes a $BA$ where each thread $T_k$ has its own
entry which then points to the thread's private subgoal frame. Each
private subgoal frame includes an extra field which is a back pointer
to the common subgoal entry. This is important in order to keep
unaltered all the tabling data structures that access subgoal
frames. To access the private information, there is no extra cost (we
still use a direct pointer), and only for the common information on
the subgoal entry we pay the extra cost of following an indirect
pointer.

Comparing with the NS and SS designs, the FS design has two major
advantages. First, memory usage is reduced to a minimum. The only
memory overhead, when compared with a single threaded evaluation, is
the $BA$ associated with each subgoal entry, and apart from the split
on the subgoal frame data structure, all the remaining structures
remain unchanged. Second, since threads are sharing the same
$AT(P_{i.j})$ structures, answers inserted by a thread for a
particular subgoal call are automatically made available to all other
threads when they call the same subgoal.

Given an arbitrary number of $NT$ running threads and assuming that
all threads have completely evaluated the same number $NC(P_i)$ of
tabled subgoal calls, Eq.~\ref{equation_mu_fs} shows the memory usage
for a predicate $P_i$ in the FS design ($MU_{FS}(P_i)$).

\begin{equation}
\begin{aligned}
& MU_{FS}(P_i) = TE + ST(P_i) + \sum\limits^{NC(P_i)}_{j=1} [SE_{FS} + BA + NT * [SF_{FS} + BP] + AT(P_{i.j})] \\
& where~~
SE_{FS} + SF_{FS} = SF
\end{aligned}
\label{equation_mu_fs}
\end{equation} 

The memory usage for the FS design is given by the sum of the memory
size of the $TE$ and $ST(P_i)$ data structures plus the summation, for
each subgoal call, of the memory used by the subgoal entry data
structure ($SE_{FS}$), the $BA$ and the $AT(P_{i.j})$ structures
added with the sizes of the private data structures of each thread
multiplied by the $NT$ threads. The private data structures of each
thread include the subgoal frame ($SF_{FS}$) and the back pointer
($BP$). The memory size of the original $SF$ is now given by the size
of the $SE_{FS}$ and $SF_{FS}$ data structures. The memory size of the
remaining structures is the same as in Yap's original table space
organization.

Since the FS design is a refinement of the SS design, next we use
Thm.~\ref{theorem_SS_vs_FS} to show that the FS design always requires
less memory than the SS design for more than one thread.

\begin{theorem}
\label{theorem_SS_vs_FS}
If $NT > 1$ and $NC(P_i) \geq 1$ then $MU_{FS}(P_i) < MU_{SS}(P_i)$.
\end{theorem}

Remember from the previous subsection that the SS behavior depends on
the amount of memory spent in the $BA$. The FS maintains this
dependency, since this structure is co-allocated inside the subgoal
entry structure. The difference between both designs occurs in the
memory usage spent in the subgoal frames and in the answer tries. For
the subgoal frames, the difference is that the size of the private
subgoal frames used by the FS design, including the back pointer, is
lower that the ones used by the SS design. For the answer trie
structures, the FS design simply does not allocate as many of these
structures has the SS design. For one thread ($NT=1$), the following
corollary can be derived from Thm.~\ref{theorem_SS_vs_FS}:

\begin{corollary}
If $NT=1$ and $NC(P_i) \geq 1$ then $ MU_{FS}(P_i) > MU_{SS}(P_i)$.
\end{corollary}

In summary, for one thread, the FS design is always worse than the SS
design and the difference increases proportionally to the number of
subgoal calls. For a number of threads higher than one, the FS design
always performs better than the SS design and the difference increases
as the number of threads and the number of subgoal calls also
increases.


\subsection{Partial Answer Sharing Design}

In the SS design, the subgoal trie structures are shared among threads
but the answers for the subgoal calls are stored in private answer
trie structures to each thread. As a consequence, no sharing of
answers between threads is done. The \emph{Partial Answer Sharing
  (PAS)} design~\cite{areias-jss16} extends the SS design to allow
threads to share answers. Threads still view their answer tries as
private but are able to consume answers from completed answer tries
computed by other threads. The idea is as follows. Whenever a thread
calls a new tabled subgoal, first it searches the table space to
lookup if any other thread has already computed the answers for that
subgoal. If so, then the thread reuses the available answers, thus
avoiding recomputing the subgoal call from scratch. Otherwise, it
computes the subgoal itself. Several threads can work on the same
subgoal call simultaneously, i.e., we do not protect a subgoal from
further evaluations while other threads have picked it up already. The
first thread completing a subgoal, shares the results by making them
available (public) to the other
threads. Figure~\ref{fig_table_space_answer_sharing} illustrates the
table space organization for the PAS design.

\begin{wrapfigure}{R}{7.5cm}
\includegraphics[width=7.5cm]{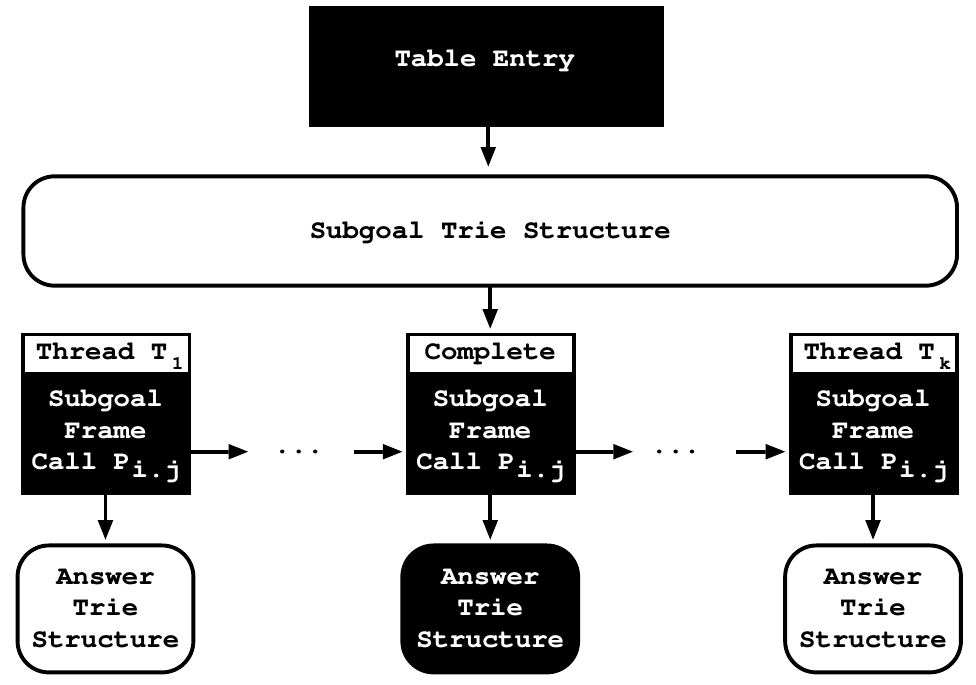}
\caption{Table space organization for the PAS design}
\label{fig_table_space_answer_sharing}
\vspace{-\intextsep}
\end{wrapfigure}

As for the SS design, threads can concurrently access the subgoal trie
structures for both read and write operations, but for the answer trie
structures, they are only concurrently accessed for reading after
completion. All subgoal frames and answer tries are initially private
to a thread. Later, when the first subgoal frame is completed, i.e.,
when we have found the full set of answers for it, it is marked as
completed (black answer trie in
Fig.~\ref{fig_table_space_answer_sharing}) and put in the beginning of
the list of private subgoal frames (configuration shown in
Fig.~\ref{fig_table_space_answer_sharing}). With the PAS design, we
also aim to improve the memory usage of the SS design by removing the
$BA$ data structure. This is a direct consequence of the analysis made
in Eq.~\ref{equation_mu_ss} where we have shown that the performance
of the SS design is directly affected by the size of the memory used
by the $BA$ structures. Thus, instead of pointing to a $BA$ as in the
SS design, now the leaf data structure in each subgoal trie path
points to a list of private subgoal frames corresponding to the
threads evaluating the subgoal call. In order to find the subgoal
frame corresponding to a thread, we may have to pay an extra cost for
navigating in the list but, once a subgoal frame is completed, we can
access it immediately since it is always stored in the beginning of
the list.

Given an arbitrary number of $NT$ running threads and assuming that
all threads have completely evaluated the same number $NC(P_i)$ of
tabled subgoal calls, Eq.~\ref{equation_mu_as} shows the memory usage
for a predicate $P_i$ in the PAS design ($MU_{PAS}(P_i)$).

\begin{equation}
\begin{aligned}
& MU_{PAS}(P_i) = TE + ST(P_i) + \sum\limits^{NC(P_i)}_{j=1} [NT(P_{i.j}) * [SF + AT(P_{i.j})]] \\ 
& where~~
NT(P_{i.j}) \leq NT
\end{aligned}
\label{equation_mu_as}
\end{equation} 

The memory usage for the PAS design is given by the sum of the memory
size of the $TE$ and $ST(P_i)$ data structures plus the summation, for
each subgoal call, of the memory used by the private structures of
each thread multiplied by $NT(P_{i.j})$ threads, where $NT(P_{i.j})$
is the number of threads evaluating the subgoal call $P_{i.j}$ in a
private fashion. Note that $NT(P_{i.j}) \leq NT$, since the threads
consuming answers from completed subgoal frames do not allocate any
extra memory. The memory size of each particular data structure is the
same as in Yap’s original table space organization.

In summary, if comparing Eq.~\ref{equation_mu_ss} with
Eq.~\ref{equation_mu_as}, we can observe that the total memory usage
of the PAS design is always less than the total memory usage of the SS
design. Additionally, we can optimize even further this design and
allow threads to delete their private $SF$ and $AT(P_{i.j})$
structures when completing, if another thread has made public its
completed subgoal frame first. With this optimization, we can end in
practice with a single $SF$ and $AT(P_{i.j})$ structure for each
subgoal call $P_{i.j}$.

If comparing with the FS design, because we only share completed
answer tries, we also avoid some problems present in the FS
design. First, we avoid the problem of dealing with concurrent updates
to the answer tries. Second, we avoid the problem of dealing with
concurrent deletes, as in the case of using mode-directed
tabling. Since the PAS design keeps the answer tries private to each
thread, the deletion of nodes can be done without any complex
machinery to deal with concurrent delete operations. Third, we avoid
the problem of managing the different set of answers that each thread
has found. As we will see in the next subsection, this can be a
problem for batched scheduling evaluation.


\subsection{Private Answer Chaining Design}

During tabled execution, there are several points where we may have to
choose between continuing forward execution, backtracking, consuming
answers from the tables or completing subgoals. The decision about the
evaluation flow is determined by the \emph{scheduling
  strategy}. Different strategies may have a significant impact on
performance, and may lead to a different ordering of solutions to the
query goal. Arguably, the two most successful tabling scheduling
strategies are \emph{local scheduling} and \emph{batched
  scheduling}~\cite{Freire-96}.

Local scheduling tries to complete subgoals as soon as possible. When
new answers are found, they are added to the table space and the
evaluation fails. Local scheduling has the advantage of minimizing the
size of \emph{clusters of dependent subgoals}. However, it delays
propagation of answers and requires the complete evaluation of the
search space.

Batched scheduling tries to delay the need to move around the search
tree by batching the return of answers to consumer subgoals. When new
answers are found for a particular tabled subgoal, they are added to
the table space and the evaluation continues. Batched scheduling can
be a useful strategy in problems which require an eager propagation
of answers and/or do not require the complete set of answers to be
found.

With the FS design, all tables are shared. Thus, since several threads
can be inserting answers in the same answer trie, when an answer
already exists, it is not possible to determine if the answer is new
or repeated for a particular thread without further support. For local
scheduling, this is not a problem since, for repeated and new answers,
local scheduling always fails. The problem occurs with batched
scheduling that requires that only the repeated answers should
fail. Threads have then to detect, during batched evaluation, whether
an answer is new and must be propagated or whether an answer is
repeated and the evaluation must fail. The \emph{Private Answer
  Chaining (PAC)} design~\cite{areias-slate15-post} extends the FS
design to keep track of the answers that were already found and
propagated per thread and subgoal
call. Figure~\ref{fig_table_space_answer_chaining_overview}
illustrates PAC's key idea.

\begin{wrapfigure}{R}{6.25cm}
\includegraphics[width=6.25cm]{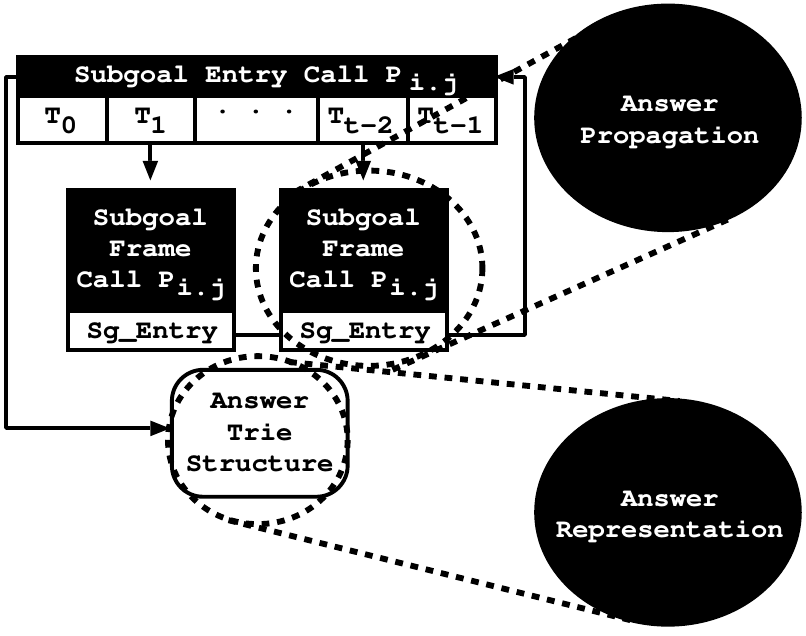}
\caption{PAC overview}
\label{fig_table_space_answer_chaining_overview}
\vspace{-\intextsep}
\end{wrapfigure}

In a nutshell, PAC splits \emph{answer propagation} from \emph{answer
  representation}, and allows the former to be privately stored in the
subgoal frame data structure of each thread, and the latter to be kept
publicly shared among threads in the answer trie data structure. This
is similar to the idea proposed by Costa and Rocha~\cite{CostaJ-09b}
for the \emph{global trie data structure}, where answers are
represented only once on a global trie and then each subgoal call has
private pointers to its set of answers. With PAC, we follow the same
key idea of representing only once each answer (as given by the FS
design), but now since we are in a concurrent environment, we use a
private chain of answers per thread to represent the answers for each
subgoal call. Later, if a thread completes a subgoal call, its PAC is
made public so that from that point on all threads can use that chain
in complete (only reading)
mode. Figure~\ref{fig_table_space_answer_chaining} illustrates the new
data structures involved in the implementation of PAC's design for a
situation where different threads are evaluating the same tabled
subgoal call $P_{i.j}$.

\begin{figure}[!ht]
\centering
\includegraphics[width=0.75\textwidth]{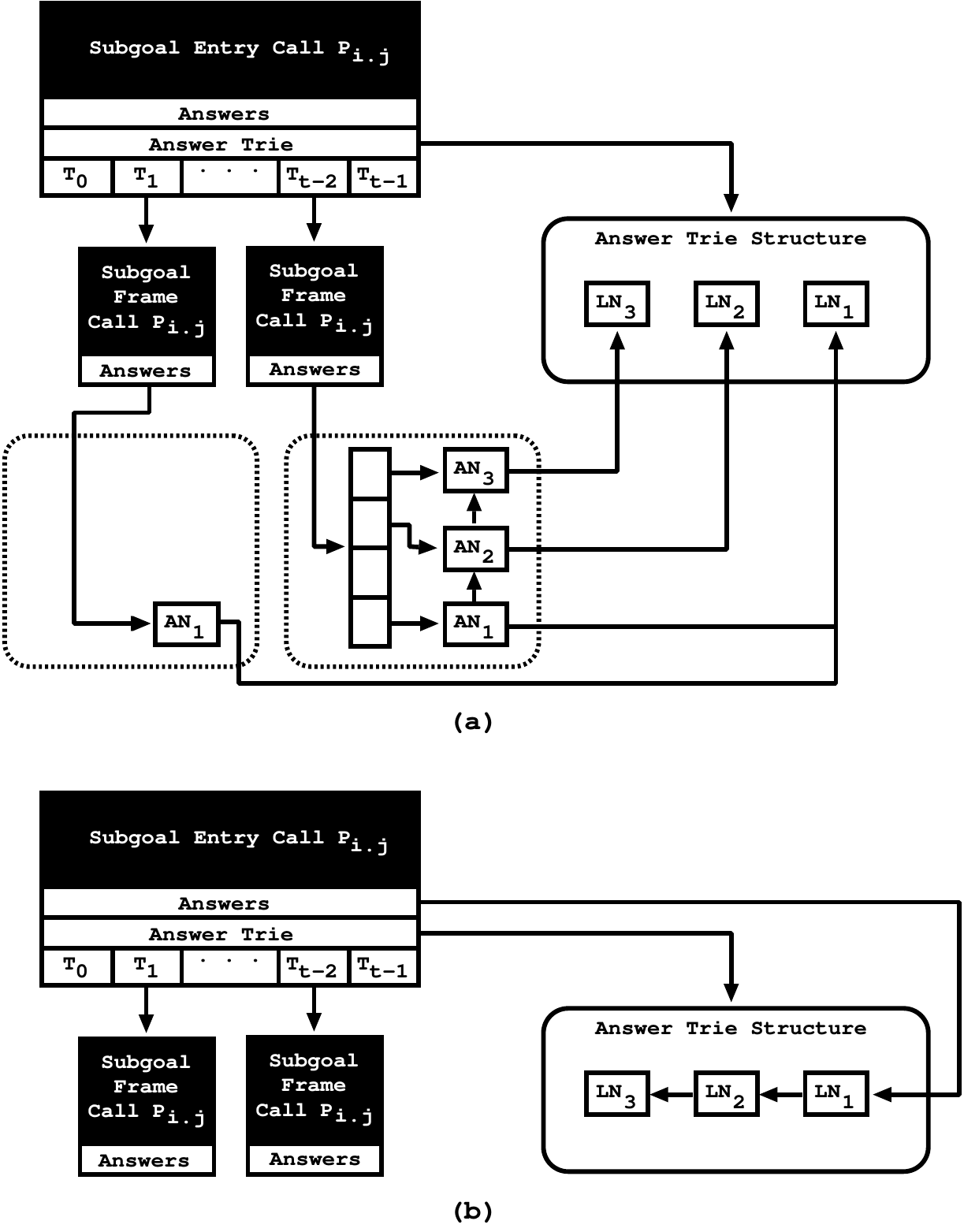}
\caption{PAC's data structures for (a) private and (b) public
  chaining}
\label{fig_table_space_answer_chaining}
\end{figure}

Figure~\ref{fig_table_space_answer_chaining}(a) shows then a situation
where two threads, $T_1$ and $T_{t-2}$, are sharing the same subgoal
entry for a call $P_{i.j}$ still under evaluation, i.e., still not yet
completed. The current state of the evaluation shows an answer trie
with 3 answers found for $P_{i.j}$. For the sake of simplicity, we are
omitting the internal answer trie nodes and we are only showing the
leaf nodes $LN_1$, $LN_2$ and $LN_3$ of each answer.

With the PAC design, the leaf nodes are not chained in the answer trie
data structure, as usual. Now, the chaining process is done privately,
and for that, we use the subgoal frame structure of each thread. On
the subgoal frame structure we added a new field, called
\emph{Answers}, to store the answers found within the execution of the
thread. In order to minimize PAC's impact, each answer node in the
private chaining has only two fields: (i) an entry pointer, which
points to the corresponding leaf node in the answer trie data
structure; and (ii) a next pointer to chain the nodes in the private
chaining. To maintain good performance, when the number of answer
nodes exceeds a certain threshold, we use a hash trie mechanism design
similar to the one presented in~\cite{Areias-ijpp15}, but without
concurrency support, since this mechanism is private to each thread.

PAC's data structures in Fig.~\ref{fig_table_space_answer_chaining}(a)
represent then two different situations. Thread $T_1$ has only found
one answer and it is using a direct answer chaining to access the leaf
node $LN_1$. Thread $T_{t-2}$ has already found three answers for
$P_{i.j}$ and it is using the hash trie mechanism within its private
chaining. In the hash trie mechanism, the answer nodes are still
chained between themselves, thus that repeated calls belonging to
thread $T_{t-2}$ can consume the answers as in the original mechanism.

Figure~\ref{fig_table_space_answer_chaining}(b) shows the state of the
subgoal call after completion. When a thread $T$ completes a subgoal
call, it frees its private consumer structures, but before doing that,
it checks whether another thread as already marked the subgoal as
completed. If no other thread has done that, then thread $T$ not only
follows its private chaining mechanism, as it would for freeing its
private nodes, but also follows the pointers to the answer trie leaf
nodes in order to create a chain inside the answer trie. Since this
procedure is done inside a critical region, no more than one thread
can be doing this chaining process. Thus, in
Fig.~\ref{fig_table_space_answer_chaining}(b), we are showing the
situation where the subgoal call $P_{i.j}$ is completed and both
threads $T_1$ and $T_{t-2}$ have already chained the leaf nodes inside
the answer trie and removed their private chaining structures.


\section{Engine Components}
\label{sec_engine_components}

This section discusses the most important engine components required
to support concurrent tabled evaluation.


\subsection{Fixed-Size Memory Allocator}

A critical component in the implementation of an efficient concurrent
tabling system is the memory allocator. Conceptually, there are two
categories of memory allocators: \emph{kernel-level} and
\emph{user-level} memory allocators. Kernel-level memory allocators
are responsible for managing memory inside the protected
sub-systems/resources of the operating system, while user-level memory
allocators are responsible for managing the \emph{heap} area, which is
the area inside the addressing space of each process where the dynamic
allocation of memory is directly done.

Evidence of the importance of a \emph{User-level Memory Allocator
  (UMA)} comes from the wide array of UMA replacement packages that
are currently available. Some examples are the
PtMalloc~\cite{ptmalloc}, Hoard~\cite{Berger-00},
TcMalloc~\cite{tcmalloc} and JeMalloc~\cite{Evans-06} memory
allocators. Many UMA subsystems were written in a time when
multiprocessor systems were rare. They used memory efficiently but
were highly serial, constituting an obstacle to the throughput of
concurrent applications, which require some form of synchronization to
protect the heap. Additionally, when a concurrent application is ran
in a multiprocessor system, other problems can occur, such \emph{heap
  blowup}, \emph{false sharing} or \emph{memory
  contention}~\cite{Masmano-06,Gidenstam-10}.

Since tabling also demands the multiple allocation and deallocation of
different sized chunks of memory, memory management plays an important
role in the efficiency of a concurrent tabling system. To satisfy this
demand, we have designed a \emph{fixed-size UMA} especially aimed for
an environment with the characteristics of concurrent
tabling~\cite{Areias-12b}. In a nutshell, fixed-size UMA separates
local and shared memory allocation, and uses local and global heaps
with pages that are formatted in blocks with the sizes of the existing
data structures. The page formatting in blocks contributes to avoid
inducing false-sharing, because different threads in different
processors do not share the same cache lines, and to avoid the heap
blowup problem, because pages migrate between local and global heaps.


At the implementation level, our proposal has local and global heaps
with pages formatted for each object type. In addition, global and
local heaps can hold free (unformatted) pages for use when a local
heap runs empty. Since modern computer architectures use pages to
handle memory, we adopted an allocation scheme based also on pages,
where each memory page only contains data structures of the same
type. In order to split memory among different threads, in our
proposal, a page can be considered a \emph{local page}, if owned by a
particular thread, or a \emph{global page},
otherwise. Figure~\ref{fig_page_based_memory} gives an overview of
this organization.


\begin{figure}[ht]
\centering
\includegraphics[width=\textwidth]{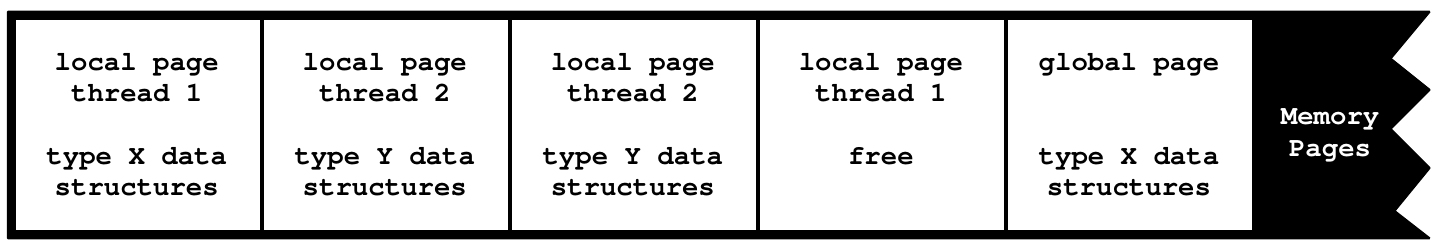}
\caption{Using pages as the basis for the fixed-size memory allocator}
\label{fig_page_based_memory}
\end{figure}

A thread can own any number of pages of the same type, of different
types and/or free pages. Any type of page (including free pages) can
be local to a thread or global, and each particular page only contains
data structures of the same type. When a page $P$ is made local to a
thread $T$, it means that $T$ gains exclusive permission to
allocate and deallocate data structures on $P$. On the other hand,
global pages have no owners and, thus, they are free from
allocate/deallocate operations. To allocate/deallocate data structures
on global pages, first the corresponding pages should be moved to a
particular thread. All running threads can access (for read or write
operations) the data structures allocated on a page, independently of
being a local or global page.

Allocating and freeing data structures are constant-time operations,
because they require only moving a structure to or from a list of free
structures. Whenever a thread $T$ requests to allocate memory for a
data structure of type $S$, it can instantly satisfy the request by
returning the first unused slot on the first available local page with
type $S$. Deallocation of a data structure of type $S$ does not free
up the memory, but only opens an unused slot on the chain of available
local pages for type $S$. Further requests to allocate memory of type
$S$ will later return the now unused memory slot. When all data
structures in a page are unused, the page is moved to the chain of
free local pages. A free local page can be reassigned later to a
different data type. When a thread $T$ runs out of available free
local pages, it must synchronize with the other threads in order to
access the global pages or the operating system's memory allocator, if
no free global page exists. This process eliminates the need to search
for suitable memory space and greatly alleviates memory
fragmentation. The only wasted space is the unused portion at the end
of a page when it cannot fit exactly with the size of the
corresponding data structures.

When a thread finishes execution, it deallocates all its private data
structures and then moves its local pages to the corresponding global
page entries. Shared structures are only deallocated when the last
running thread (usually the main thread) abolishes the tables. Thus,
if a thread $T$ allocates a data structure $D$, then it will be also
responsible for deallocating $D$, if $D$ is private to $T$, or $D$
will remain live in the tables, if $D$ is shared, even when $T$ finish
execution. In the latter case, $D$ can be only deallocated by the last
running thread $L$. In such case, $D$ is made to be local to $L$ and
the deallocation process follows as usual.


\subsection{Lock-Free Data Structures}

Another critical component in the implementation of an efficient
concurrent tabling system is the design of the data structures and
algorithms that manipulate shared tabled data. As discussed before,
Yap's table space follows a two-level trie data structure, where one
level stores the tabled subgoal calls and the other stores the
computed answers. Depending on the number of subgoal calls or answers,
the paths inside a trie, corresponding to the subgoal calls or
answers, might have several trie nodes per internal level of the trie
structure. Whenever an internal trie level becomes saturated, a
\emph{hash mechanism} is used to provide direct node access and
therefore optimize the search for the data within the trie
level. Figure~\ref{fig_trie_hash_overview} shows a hashing mechanism
for an internal trie level within the subgoal and answer data
structures.
 
\begin{wrapfigure}{R}{8.5cm}
\includegraphics[width=8.5cm]{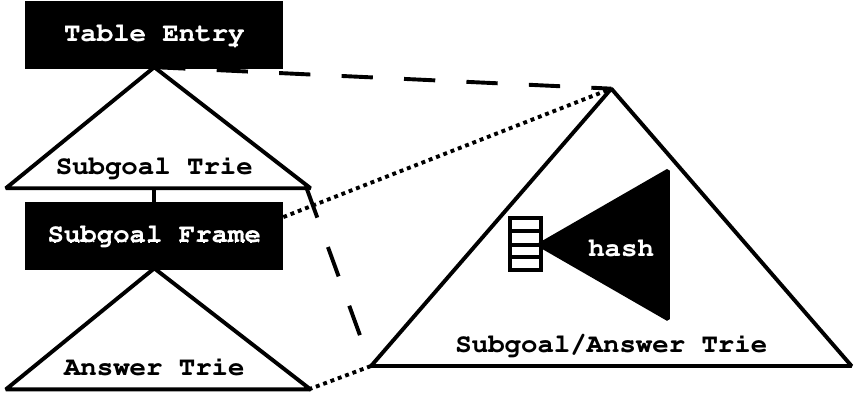}
\caption{The hashing mechanism within a trie level}
\label{fig_trie_hash_overview}
\vspace{-\intextsep}
\end{wrapfigure}

Several approaches for hashing mechanisms exist. The most important
aspect of a hashing mechanism is its behavior in terms of hash
collisions, i.e., when two keys collide and occupy the same hash table
location. Multiple solutions exist that address the collision
problem. Among these are the \emph{open addressing} and \emph{closed
  addressing} approaches~\cite{Tenenbaum-90,Knuth-98}. In open
addressing, the hash table stores the objects directly within the hash
table internal array, while in closed addressing, every object is
stored directly at an index in the hash table's internal array. In
closed addressing, collisions are solved by using other arrays or
linked lists. Yap's tabling engine uses \emph{separate
  chaining}~\cite{Knuth-98} to solve hash collisions. In the separate
chaining mechanism, the hash table is implemented as an array of
linked lists. The basic idea of separate chaining techniques is to
apply linked lists for collision management, thus in case of a
conflict a new object is appended to the linked list.

Our initial approach to deal with concurrency within the trie
structures was to use \emph{lock-based
  strategies}~\cite{Areias-12a}. However, lock-based data structures
have their performance restrained by multiple problems, such as,
convoying, low fault tolerance and delays occurred inside a critical
region. We thus shifted our attention in to taking advantage of the
low-level \emph{Compare-And-Swap (CAS)} operation, that nowadays can
be widely found on many common architectures. The CAS operation is an
\emph{atomic instruction} that compares the contents of a memory
location to a given value and, if they are the same, updates the
contents of that memory location to a given new value. The CAS
operation is at the heart of many \emph{lock-free} (also known as
non-blocking) data structures~\cite{Herlihy-87}. Non-blocking data
structures offer several advantages over their blocking counterparts,
such as being immune to deadlocks, lock convoying and priority
inversion, and being preemption tolerant, which ensures similar
performance regardless of the thread scheduling policy. Using
lock-free techniques, we have created two proposals for concurrent
hashing data structures especially aimed to be as effective as
possible in a concurrent tabling engine and without introducing
significant overheads in the sequential
execution. Figure~\ref{fig_trie_hash_proposals} shows the architecture of
the two proposals.

\begin{wrapfigure}{R}{7.5cm}
\includegraphics[width=7.5cm]{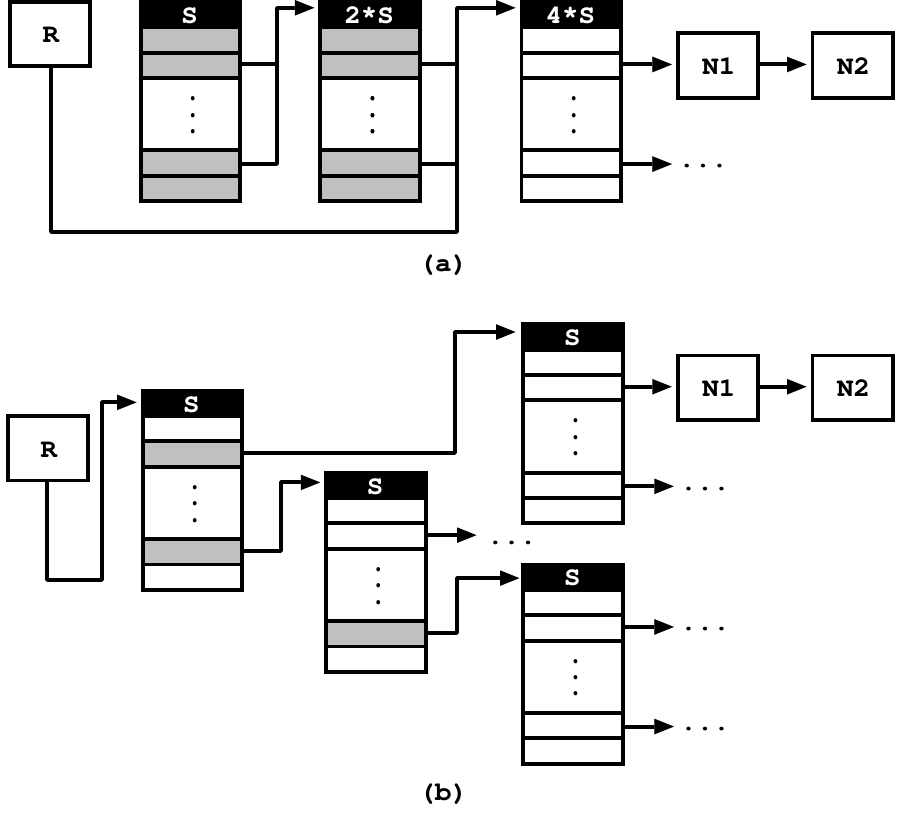}
\caption{Architecture of the two lock-free hash proposals}
\label{fig_trie_hash_proposals}
\vspace{-\intextsep}
\end{wrapfigure}

Both proposals include a root node $R$ and have a hashing mechanism
composed by a bucket array and a hash function that maps the nodes
into the entries in the bucket array. The first proposal, shown in
Fig.~\ref{fig_trie_hash_proposals}(a), implements a dynamic resizing
of the hash tables by doubling the size of the bucket array whenever
it becomes saturated~\cite{Areias-14}. It starts with an initial
bucket array with $S$ entries and, whenever the hash bucket array
becomes saturated, i.e., when the number of nodes in a bucket entry
exceeds a pre-defined threshold value and the total number of nodes
exceeds $S$, then the bucket array is expanded to a new one with $2*S$
entries. This expansion mechanism is executed by a single thread,
meaning that no more than one expansion can be done at a time. If the
thread executing the expansion suspends for some reason (for example,
be suspended by the operating system scheduler), then all the
remaining threads can still be searching and inserting nodes in the
trie level that is being expanded in a lock-free fashion, but no other
thread will be able to expand the same trie level. When the process of
bucket expansion is completed for all $S$ bucket entries, node $R$ is
updated to refer to the new bucket array with $2*S$ entries. Since the
size of the hashes doubles on each expansion, this proposal is highly
inappropriate to be integrated with the fixed-size UMA.

The second proposal, shown in Fig.~\ref{fig_trie_hash_proposals}(b),
was designed to be compatible with the fixed-size UMA. It is based on
\emph{hash tries data structures} and is aimed to be a simpler and
more efficient lock-free proposal that disperses the synchronization
regions as much as possible in order to minimize problems such as
false sharing or cache memory ping pong
effects~\cite{Areias-ijpp15}. Hash tries (or hash array mapped tries)
are another trie-based data structure with nearly ideal
characteristics for the implementation of hash
tables~\cite{Bagwell-01}. As shown in
Fig.~\ref{fig_trie_hash_overview}, in this proposal, we still have the
original subgoal/answer trie data structures which include a hashing
mechanism whenever an internal trie level becomes saturated, but now
the hashing mechanism is implemented using hash tries data structures.

An essential property of the trie data structure is that common
prefixes are stored only once~\cite{Fredkin-62}, which in the context
of hash tables allows us to efficiently solve the problems of setting
the size of the initial hash table and of dynamically resizing it in
order to deal with hash collisions. In a nutshell, a hash trie is
composed by \emph{internal hash arrays} and \emph{leaf nodes} (nodes
$N1$ and $N2$ in Fig.~\ref{fig_trie_hash_proposals}(b)) and the
internal hash arrays implement a hierarchy of hash levels of fixed
size $S=2^w$. To map a node into this hierarchy, first we compute the
hash value $h$ and then we use chunks of $w$ bits from $h$ to index
the entry in the appropriate hash level. Hash collisions are solved by
simply walking down the tree as we consume successive chunks of $w$
bits from the hash value $h$. Whenever a hash bucket array becomes
saturated, i.e., when the number of nodes in a bucket entry exceeds a
pre-defined threshold value, then the bucket array is expanded to a
new one with $S$ entries. As for the previous proposal, this expansion
mechanism is executed by a single thread. If the thread executing the
expansion suspends for some reason, then all the remaining threads can
still be searching and inserting nodes in the bucket entry in a
lock-free fashion. Compared with the previous proposal, this proposal
has a fined grain synchronization region, because it blocks only one
bucket entry per expansion.


\section{Performance Analysis}


Our work on combining tabling with parallelism started some years ago
when the first approach for implicit parallel tabling was
presented~\cite{Rocha-99a}. Such approach lead to the design and
implementation of an or-parallel tabling system, named
OPTYap~\cite{Rocha-01}. In OPTYap, each worker behaves like a
sequential tabling engine that fully implements all the tabling
operations. During the evaluation, the or-parallel component of the
system is triggered to allow synchronized access to the table space
and to the common parts of the search tree, or to schedule workers
running out of alternatives to exploit.

OPTYap has shown promising results in several tabled
benchmarks~\cite{Rocha-01}. The worst results were obtained in the
transitive closure of the right recursive definition of the path
problem using a grid configuration, where no speedups were obtained
with multiple workers. The bad results achieved in this benchmark were
explained by the higher rate of contention in Yap's internal data
structures, namely in the subgoal frames. A closer analysis showed
that the number of suspension/resumptions operations is approximately
constant with the increase in the number of workers, thus suggesting
that there are answers that can only be found when other answers are
also found, and that the process of finding such answers cannot be
anticipated. In consequence, suspended branches have always to be
resumed to consume the answers that could not be found sooner.

More recently, we shifted our research towards explicit parallelism
specially aimed for multithreaded environments. Initial results were
promising as we were able to significantly reduce the contention for
concurrent table accesses~\cite{Areias-12a,Areias-12b}. Later, we
presented first speedup results for the right recursive definition of
the path problem, using a naive multithreaded scheduler that considers
a set of different starting points (queries) in the graph to be run by
a set of different threads. In this work, Yap obtained a maximum
speedup of $10.24$ for 16 threads~\cite{Areias-ijpp15}. Although,
these results were better than the ones presented earlier for implicit
parallelism, they were mostly due to the different scheduler strategy
adopted to evaluate the benchmark. On the other hand, such work also
showed that with 32 threads, no improvements were obtained compared
with 16 threads. A closer analysis showed again that such behavior was
related with the large number of subgoal call dependencies in the
program. We thus believe that the ordering to which the answers are
found in some problems, like in the evaluation of the transitive
closure of strongly connected graphs, is a major problem that
restricts concurrency/parallelism in tabled programs.



In what follows, we start with worst case scenarios to study how
independent flows of execution running simultaneously interfere at the
low-level engine. Next, we focus on two well-known dynamic programming
problems, the Knapsack and LCS problems, and we discuss how we were
able to scale their execution by using Yap's multithreaded tabling
engine. The environment of our experiments was a machine with 32-Core
AMD Opteron (TM) Processor 6274 (2 sockets with 16 cores each) with
32GB of main memory, running the Linux kernel 3.16.7-200.fc20.x86 64
with Yap Prolog 6.3\footnote{Available at
  \url{https://github.com/miar/yap-6.3}.}.



\subsection{Experiments on Worst Case Scenarios}

We begin with experimental results for concurrent tabled evaluation
using local and batched scheduling with the NS, SS and PAC designs for
worst case scenarios that stress the trie data structures. For the
sake of simplicity, we will present only the best results, which were
always achieved when using the fixed-size UMA and the second lock-free
proposal. We do not show results for the CS and PAS designs because
they are not meaningful in this context, as we will see next. The
results for the FS design are identical to PAC's results, except for
batched scheduling which FS does not support.

For benchmarking, we used the set of tabling benchmarks
from~\cite{Areias-12b} which includes 19 different programs in
total. We choose these benchmarks because they have characteristics
that cover a wide number of scenarios in terms of trie usage. They
create different trie configurations with lower and higher number of
nodes and depths, and also have different demands in terms of trie
traversing\footnote{We show a more detailed characterization of the
  benchmark set in~\ref{appendix_bechmark_details}.}.


To create worst case scenarios that stress the table data structures,
\emph{we ran all threads starting with the same query goal}. By doing
this, it is expected that threads will access the table space, to
check/insert for subgoals and answers at similar times, thus causing a
huge stress on the same critical regions. In particular, for this set
of benchmarks, this will be especially the case for the answer tries,
since the number of answers clearly exceeds the number of
subgoals. Please note that, despite all threads are executing the same
program they have independent flows of execution, i.e., we are not
trying to parallelize the execution, but study how independent flows
of execution (in this case, identical flows of execution) interfere at
the low-level engine. By focusing first on the worst case scenarios,
we can infer the \emph{highest overhead ratios when compared with one
  thread} (or the lowest bounds of performance) that each design might
have when used with multiple threads in other real world
applications. For each table design, there are two main sources of
overheads: (i) the synchronization required to interact with the
memory allocator, which is proportional to the memory consumption
bounds discussed in Section~\ref{sec_concurrent_table_designs}; and
(ii) the synchronization required to interact with the table space,
which is proportional to the number of data structures that can be
accessed concurrently in each design. The overheads originated from
these two sources are not easy to isolate in order to evaluate the
weight of each in the execution time. The design of the memory
allocator clearly plays an important role in the former source of
overhead and the use of lock-free data structures is important to
soften the weight of the latter.

Table~\ref{tab_batched_overhead} shows the overhead ratios, when
compared with the NS design with 1 thread (running with local
scheduling and without the fixed-size UMA) for the NS, SS and PAC
designs running 1, 8, 16, 24 and 32 threads with local and batched
scheduling on the set of benchmarks. In order to give a fair weight to
each benchmark, the overhead ratio is calculated as follows. We begin
by running ten times each benchmark $B$ for each design $D$ with $T$
threads. Then, we calculate the average of those ten runs and use that
value ($D_{BT}$) to put it in perspective against the base time, which
is the average of the ten runs of the NS design with one thread
($NS_{B1}$)\footnote{The base times for the NS design are presented in
  Table~\ref{tab_benchs} in~\ref{appendix_bechmark_details}.}. For
that, we use the following formula for the overhead $O_{DBT} = D_{BT}
/ NS_{B1}$. After calculating all the overheads $O_{DBT}$ for a
certain design $D$ and number of threads $T$ corresponding to the
several benchmarks $B$, we calculate the respective minimum, average,
maximum and standard deviation overhead ratios. The higher the
overhead, the worse the design behaves. An overhead of 1.00 means that
the design behaves similarly to the base case and is thus immune to
the fact of having other execution flows running simultaneously.


\begin{table}[t]
\centering
\caption{Overhead ratios, when compared with the NS design with 1
  thread (running with local scheduling and without the fixed-size
  UMA), for the NS, SS and PAC designs running 1, 8, 16, 24 and 32
  threads with local and batched scheduling (best ratios by row and by
  design for the Minimum, Average and Maximum are in bold)}
\begin{tabular}{ll|cc|cc|cc}
\multicolumn{2}{l|}{\multirow{2}{*}{\bf Threads}} &
\multicolumn{2}{c|}{\multirow{1}{*}{\bf NS}} &
\multicolumn{2}{c|}{\multirow{1}{*}{\bf SS}} & 
\multicolumn{2}{c}{\multirow{1}{*}{\bf PAC}} \\
& 
& \multicolumn{1}{c}{\bf Local}
& \multicolumn{1}{c|}{\bf Batched}
& \multicolumn{1}{c}{\bf Local}
& \multicolumn{1}{c|}{\bf Batched}
& \multicolumn{1}{c}{\bf Local}
& \multicolumn{1}{c}{\bf Batched}\\
\hline\hline
\multirow{4}{*}{\bf 1}
& {\bf Min }& {\bf 0.53}& 0.55& {\bf 0.54}& 0.55& 1.01& {\bf 0.95}\\
& {\bf Avg }& {\bf 0.78}& 0.82& {\bf 0.84}& 0.90& {\bf 1.30}& 1.46\\
& {\bf Max }& 1.06& {\bf 1.05}& {\bf 1.04}& {\bf 1.04}& {\bf 1.76}& 2.33\\
& {\bf StD }& 0.15& 0.14& 0.17& 0.16& 0.22& 0.44\\
\hline
\multirow{4}{*}{\bf 8}
& {\bf Min }& 0.66& {\bf 0.63}& 0.66& {\bf 0.63}& 1.16&{\bf  0.99}\\
& {\bf Avg }& {\bf 0.85}& 0.88& {\bf 0.92}& 0.93& {\bf 1.88}& 1.95\\
& {\bf Max }& {\bf 1.12}& 1.14& 1.20& {\bf 1.15}& {\bf 2.82}& 3.49\\
& {\bf StD }& 0.13& 0.14& 0.15& 0.14& 0.60& 0.79\\
\hline
\multirow{4}{*}{\bf 16}
& {\bf Min }& 0.85& {\bf 0.75}& 0.82& {\bf 0.77}& 1.17& {\bf 1.06}\\
& {\bf Avg }& {\bf 0.98}& 1.00& {\bf 1.04}& 1.05& {\bf 1.97}& 2.08\\
& {\bf Max }& {\bf 1.16}& 1.31& 1.31& {\bf 1.28}& {\bf 3.14}& 3.69\\
& {\bf StD }& 0.09& 0.17& 0.12& 0.13& 0.65& 0.83\\
\hline
\multirow{4}{*}{\bf 24}
& {\bf Min }& {\bf 0.91}& 0.93& 1.02& {\bf 0.98}& 1.16& {\bf 1.09}\\
& {\bf Avg }& {\bf 1.15}& 1.16& 1.22& {\bf 1.19}& {\bf 2.06}& 2.19\\
& {\bf Max }& 1.72& {\bf 1.60}& 1.81& {\bf 1.61}& {\bf 3.49}& 4.08\\
& {\bf StD }& 0.20& 0.21& 0.18& 0.16& 0.70& 0.91\\
\hline
\multirow{4}{*}{\bf 32}
& {\bf Min }& 1.05& {\bf 1.04}& {\bf 1.07}& 1.12& 1.33& {\bf 1.26}\\
& {\bf Avg }& 1.51& {\bf 1.49}& 1.54& {\bf 1.51}& {\bf 2.24}& 2.41\\
& {\bf Max }& {\bf 2.52}& 2.63& {\bf 2.52}& 2.62& {\bf 3.71}& 4.51\\
& {\bf StD }& 0.45& 0.45& 0.42& 0.43& 0.74& 1.02\\
\end{tabular}
\label{tab_batched_overhead}
\end{table}

By observing Table~\ref{tab_batched_overhead}, we can notice that for
one thread, on average, local scheduling is sightly better than
batched on the three designs. As we increase the number of threads,
one can observe that, for the NS and SS designs, both scheduling
strategies show very close minimum, average and maximum overhead
ratios. For the PAC design, the best minimum overhead ratio is always
for batched scheduling but, for the average and maximum overhead
ratio, local scheduling is always better than batched scheduling. For
the average and maximum overhead ratios, the difference between local
and batched scheduling in the PAC design is slightly higher than in
the NS and SS designs, which can be read as an indication of the
overhead that PAC introduces into the FS design. Recall that whenever
an answer is found during the evaluation, PAC requires that threads
traverse their private consumer data structures to check if the answer
was already found (and propagated).

Finally, we would like to draw the reader's attention to the worst
results obtained (the ones represented by the maximum rows). For 32
threads, the NS, SS and PAC designs have overhead results of
2.52/2.63, 2.52/2.62 and 3.71/4.51, respectively for local/batched
scheduling. These are outstanding results if we compare them with the
results obtained in our first approach~\cite{Areias-12a}, without the
fixed-size UMA and without lock-free data structures, where for local
scheduling with 24 threads, the NS, SS and FS designs had average
overhead results of 18.64, 17.72 and 5.42, and worst overhead results
of 47.89, 47.60 and 11.49, respectively. Results for the XSB Prolog
system, also presented in~\cite{Areias-12a}, for the same set of
benchmarks showed average overhead results of 6.1 and worst overhead
results of 10.31. We thus argue that the combination of a fixed-size
UMA with lock-free data structures is the best proposal to support
concurrency in general purpose multithreaded tabling applications.


\subsection{Experiments on Dynamic Programming Problems}

As mentioned in subsection~\ref{sec_implicit_explicit}, with a fully
explicit approach, it is left to the user to break the problem into
tasks for concurrent execution, assign them to the available workers
and control the execution and the synchronization points, i.e., it is
\emph{not} the tabled execution system that is responsible for doing
that, the execution system only provides the mechanisms/interface for
allowing simultaneous flows of execution. Thus, the user-level
scheduler implemented by the user, to support the division of the
problem in concurrent tasks and control the execution and
synchronization points, plays a key role in the process of trying to
obtain speedups through parallel execution. This means that we cannot
evaluate the infrastructure of a concurrent tabling engine just by
running some benchmarks if we do not put a big effort in a good
scheduler design, which is independent from such infrastructure.

In this subsection, we show how dynamic programming problems fit well
with concurrent tabled evaluation~\cite{areias-jss16}. To do so, we
used two well-known dynamic programming problems, the \emph{Knapsack}
and the \emph{Longest Common Subsequence (LCS)} problems. The Knapsack
problem~\cite{Martello-90} is a well-known problem in combinatorial
optimization that can be found in many domains such as logistics,
manufacturing, finance or telecommunications. Given a set of items,
each with a weight and a profit, the goal is to determine the number
of items of each kind to include in a collection so that the total
weight is equal or less than a given capacity and the total profit is
as much as possible. The problem of computing the length of the LCS is
representative of a class of dynamic programming algorithms for string
comparison that are based on getting a similarity degree. A good
example is the sequence alignment, which is a fundamental technique
for biologists to investigate the similarity between species.

For the Knapsack problem, we fixed the number of items and capacity,
respectively, 1,600 and 3,200. For the LCS problem, we used sequences
with a fixed size of 3,200 symbols. Then, for each problem, we created
three different datasets, D$_{10}$, D$_{30}$ and D$_{50}$, meaning
that the values for the weights/profits for the Knapsack problem and
the symbols for LCS problem where randomly generated in an interval
between 1 and 10\%, 30\% and 50\% of the total number of
items/symbols, respectively.


For both problems, we implemented either \emph{multithreaded tabled
  top-down} and \emph{multithreaded tabled bottom-up} user-level
scheduler approaches. For the top-down approaches, we followed Stivala
\emph{et al.}'s work~\cite{Stivala-10} where a set of threads solve
the entire program independently but with a randomized choice of the
sub-problems. Figure~\ref{fig_knap_top_down_eval_tree-mt} illustrates
how this was applied in the case of the Knapsack problem considering
$N$ items and $C$ capacity. A set of threads begin the execution with
the same top query tabled call, $ks(N,C)$ in
Fig.~\ref{fig_knap_top_down_eval_tree-mt}, but then, on each level of
the evaluation tree, each thread randomly decides which branch will be
evaluated first, the exclude item branch (\emph{Exc}) or the include
item branch (\emph{Inc}). This random decision is aimed to disperse
the threads through the evaluation tree\footnote{A similar strategy
  was followed for the LCS problem.}.

\begin{figure}[!ht]
\centering
\includegraphics[width=12cm]{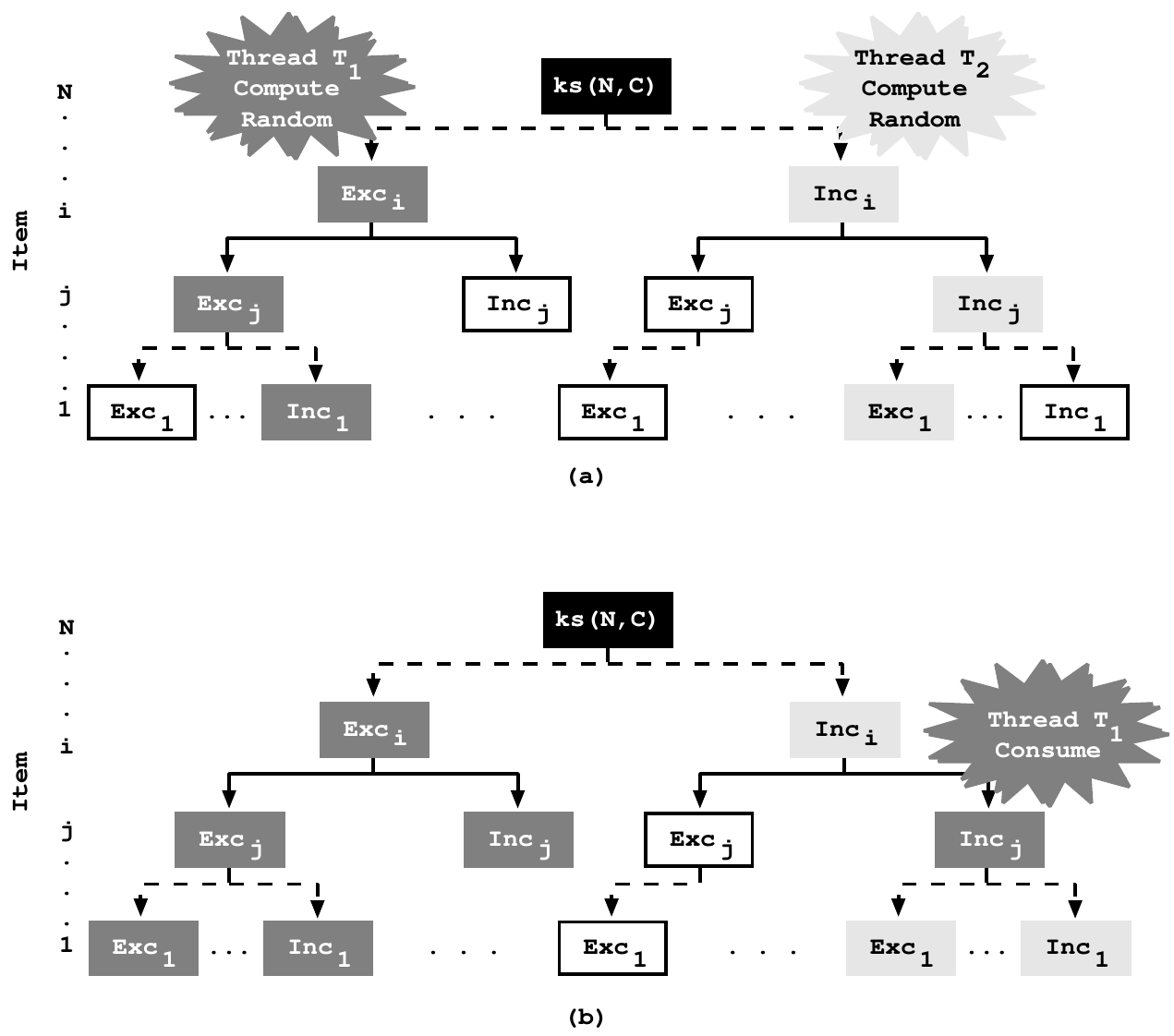}
\caption{Knapsack multithreaded tabled top-down approach}
\label{fig_knap_top_down_eval_tree-mt}
\end{figure}

Figure~\ref{fig_knap_top_down_eval_tree-mt}(a) shows a situation
where, starting from a certain item \emph{i} and capacity, thread
$T_1$ is evaluating the left branch of the tree ($Exc_i$), while
thread $T_2$ is evaluating the right branch ($Inc_i$)\footnote{For
  simplicity of presentation, the capacity values are not shown in
  Fig.~\ref{fig_knap_top_down_eval_tree-mt}. Note however that the
  tabled call corresponding to a $Exc_i$ or $Inc_i$ branch in
  different parts of the evaluation tree can be called with different
  capacity values, meaning that, in fact, they are different tabled
  calls. Only when the item and the capacity values are the same, the
  tabled call is also the same.}. Notice that although the threads are
evaluating the branches of the tree in a random order, they still have
to evaluate all branches so that they can find the optimal solution
for the Knapsack problem. So, the random decision is only about the
evaluation order of the branches and not about skipping
branches. Figure~\ref{fig_knap_top_down_eval_tree-mt}(b) shows then a
situation where thread $T_1$ has completely evaluated the $Exc_i$
branch of the tree and has moved to the $Inc_i$ branch where it is now
evaluating a $Inc_j$ branch already evaluated by thread $T_2$. Since
the result for that branch is already stored in the corresponding
table, thread $T_1$ simply consumes the result, thus avoiding its
computation.

For each sub-problem, two alternative execution choices are available:
(i) exclude first and include next, or (ii) include first and exclude
next. The randomized choice of sub-problems results in the threads
diverging to compute different sub-problems simultaneously while
reusing the sub-problem's results computed in the meantime by the
other threads. Since the number of sub-problem is usually high in
this kind of problems, it is expected that the available set of
sub-problems will be evenly divided by the number of available threads
resulting in less computation time required to reach the final result.

We have implemented two alternative versions. The first version
(YAP$_{TD_1}$) simply follows Stivala et al.'s original random
approach. The second version (YAP$_{TD_2}$) extends the first one with
an extra step where the computation is first moved forward (i.e.,
to a deeper item/symbol in the evaluation tree) using a random
displacement of the number of items/symbols (we used a $maxRandom$
value corresponding to $10\%$ of the total number of items/symbols in
the problem) and only then the computation is performed for the next
item/symbol, as usual.

For the bottom-up user-level scheduler approaches (YAP$_{BU}$), the
Knapsack version is based on~\cite{Kumar-94} and the LCS version is
based on~\cite{Kumar-02}. Figure~\ref{fig_knapsack_matrix-mt}
illustrates the case of the Knapsack problem for $N$ items and $C$
capacity. The evaluation is done bottom-up with increasing capacities
$c \in \{1,...,C\}$ until computing the maximum profit for the given
capacity $C$, which corresponds to the query goal $ks(N,C)$. The
bottom-up characteristic comes from the fact that, given a Knapsack
with capacity $c$ and using $i$ items, $i < N$, the decision to
include the next item $j$, $j=i+1$, leads to two situations: (i) if
$j$ is not included, the Knapsack profit is unchanged; (ii) if $j$ is
included, the profit is the result of the maximum profit of the
Knapsack with the same $i$ items but with capacity $c - w_j$ (the
capacity needed to include the weight $w_j$ of item $j$) increased by
$p_j$ (the profit of the item $j$ being included). The algorithm then
decides whether to include an item based on which choice leads
to maximum profit. Thus, computing a row $i$ depends only on the
sub-problems at row $i-1$. A possible parallelization is, for each
row, to divide the computation of the $C$ columns between the
available threads and then wait for all threads to complete in order
to synchronize before computing the next
row. Figure~\ref{fig_knapsack_matrix-mt}(a) shows an example with two
threads, $T_1$ and $T_2$, where the computation of the $C$ columns
within the evaluation matrix is divided in smaller chunks and each
chunk is evaluated by the same thread.

\begin{figure}[!ht]
\centering
\includegraphics[width=\textwidth]{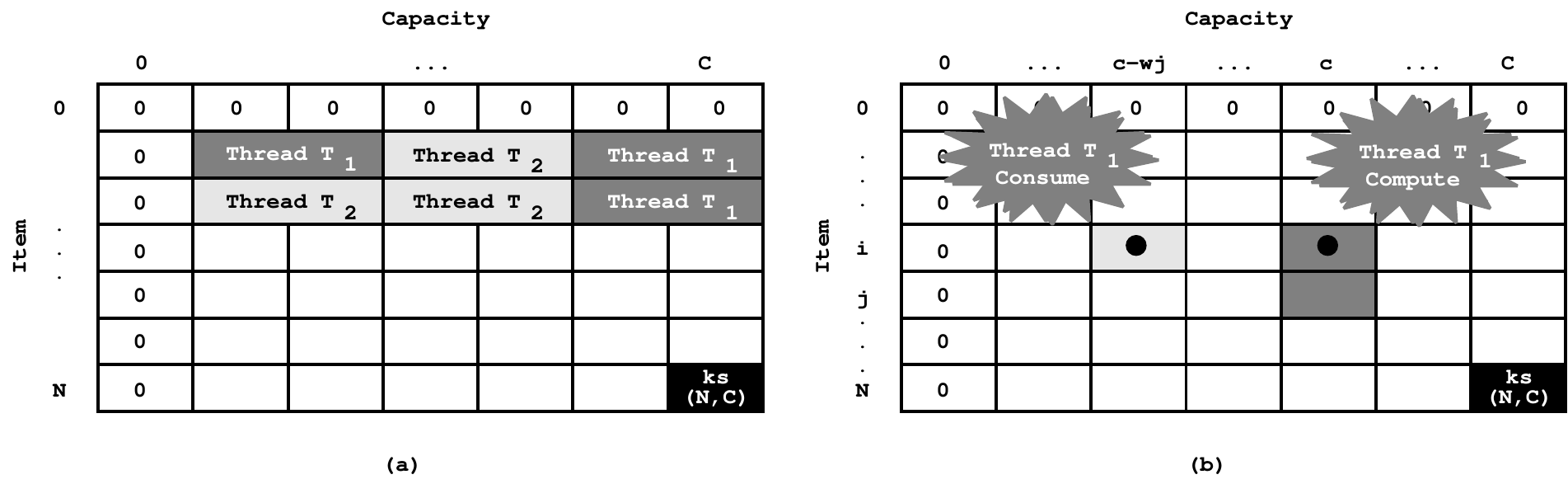}
\caption{Knapsack multithreaded tabled bottom-up approach}
\label{fig_knapsack_matrix-mt}
\end{figure}

Figure~\ref{fig_knapsack_matrix-mt}(b) shows then a situation where
the cell corresponding to call $ks(j,c)$ is being evaluated by thread
$T_1$. As explained above, this involves computing the values for
$ks(i,c-w_j)$ and $ks(i,c)$ (cells denoted with a black circle in
Fig.\ref{fig_knapsack_matrix-mt}(b)). Since we want to take advantage
of the built-in tabling mechanism, we can avoid the synchronization
between rows mentioned above. Hence, when a sub-problem in the
previous row was not computed yet (i.e., marked as completed in one of
the subgoal frames for the given call), instead of waiting for the
corresponding result to be computed by another thread, the current
thread starts also its computation and for that it can recursively
call many other sub-problems not computed yet. Despite this can lead
to redundant sub-computations, it avoids synchronization. In fact, as
we will see, this approach showed to be very effective. The situation
in Fig.~\ref{fig_knapsack_matrix-mt}(b) shows the case where thread
$T_1$ consumes the value for call $ks(i,c-w_j)$ from the tables
(already computed by $T_2$) but computes the value for $ks(i,c)$.


To evaluate the performance of the multithreaded tabled top-down and
bottom-up approaches, we used local scheduling with the PAS design,
together with the fixed-size UMA and the support for lock-free data
structures within the subgoal trie data structure. For the bottom-up
approaches, standard tabling is enough but for the top-down
approaches, mode-directed tabling is mandatory since we want to
maximize the profit, in the case of the Knapsack problem, and the
length of the longest common subsequence, in the case of the LCS
problem. To put our results in perspective, we also experimented with
XSB Prolog version 3.4.0 using the shared tables
model~\cite{Marques-08} for the bottom-up approaches (since XSB does
not support mode-directed tabling, it could not be used for the
top-down approaches).

\begin{table}[t]
\centering
\caption{Execution time, in milliseconds, for one thread (sequential
  and multithreaded version) and corresponding speedup (against one
  thread running the multithreaded version) for the execution with 8,
  16, 24 and 32 threads, for the top-down and bottom-up approaches of
  the Knapsack problem using the Yap and XSB Prolog systems}
\begin{tabular}{cc||c|c|rrrr||r}
\multicolumn{2}{c||}{\multirow{3}{*}{\bf System/Dataset}}
& {\bf Seq.}
& \multicolumn{5}{c||}{\bf \# Threads (p)} 
& {\bf Best} \\
&
& {\bf Time} 
& {\bf Time (T$_1$)}  & \multicolumn{4}{c||}{\bf Speedup (T$_1$/T$_p$)}
& {\bf Time}\\ 
&
& {\bf (T$_{seq}$)}
& {\bf 1} & {\bf 8} & {\bf 16} & {\bf 24} & {\bf 32} &  {\bf (T$_{best}$)}\\
\hline\hline
\multicolumn{8}{l}{\bf Top-Down Approaches} \\
\multicolumn{1}{c}{\multirow{3}{*}{\bf YAP$_{TD_1}$}} 
& {\bf D$_{10}$} & 14,330 & 19,316 & 1.96 & {\bf 2.12} &      2.04  & 1.95 & 9,115 \\
& {\bf D$_{30}$} & 14,725 & 19,332 & 3.57 & {\bf 4.17} &      4.06  & 3.93 & 4,639 \\
& {\bf D$_{50}$} & 14,729 & 18,857 & 4.74 &      6.28  & {\bf 6.44} & 6.41 & 2,930 \\
\hline
\multicolumn{1}{c}{\multirow{3}{*}{\bf YAP$_{TD_2}$}} 
& {\bf D$_{10}$} & 19,667 & 24,444 & 6.78 & 12.35 & 15.44 & {\bf 18.19} & 1,344 \\
& {\bf D$_{30}$} & 19,847 & 25,609 & 7.15 & 13.83 & 17.37 & {\bf 20.47} & 1,251 \\
& {\bf D$_{50}$} & 19,985 & 25,429 & 7.27 & 13.70 & 17.35 & {\bf 20.62} & 1,233 \\
\hline
\multicolumn{8}{l}{\bf Bottom-Up Approaches} \\
\multicolumn{1}{c}{\multirow{3}{*}{\bf YAP$_{BU}$}} 
& {\bf D$_{10}$} & 12,614 & 17,940 & 7.17 & 13.97 & 18.31 & {\bf 22.15} & 810 \\
& {\bf D$_{30}$} & 12,364 & 17,856 & 7.23 & 13.78 & 18.26 & {\bf 21.94} & 814 \\
& {\bf D$_{50}$} & 12,653 & 17,499 & 7.25 & 14.01 & 18.34 & {\bf 21.76} & 804 \\
\hline
\multicolumn{1}{c}{\multirow{3}{*}{\bf XSB$_{BU}$}} 
& {\bf D$_{10}$} & {\bf 32,297} & 38,965 & 0.87 & 0.66 & 0.62 & 0.55 & 32,297 \\
& {\bf D$_{30}$} & {\bf 32,063} & 38,007 & 0.86 & 0.61 & 0.56 & 0.53 & 32,063 \\
& {\bf D$_{50}$} & {\bf 31,893} & 38,534 & 0.84 & 0.58 & 0.57 & 0.57 & 31,893 \\
\end{tabular}
\label{tab_knapsack}
\end{table}

\begin{table}[t]
\centering
\caption{Execution time, in milliseconds, for one thread (sequential
  and multithreaded version) and corresponding speedup (against one
  thread running the multithreaded version) for the execution with 8,
  16, 24 and 32 threads, for the top-down and bottom-up approaches of
  the LCS problem using the Yap and XSB Prolog systems}
\begin{tabular}{cc||c|c|rrrr||r}
\multicolumn{2}{c||}{\multirow{3}{*}{\bf System/Dataset}}
& {\bf Seq.}
& \multicolumn{5}{c||}{\bf \# Threads (p)} 
& {\bf Best} \\
&
& {\bf Time} 
& {\bf Time (T$_1$)}  & \multicolumn{4}{c||}{\bf Speedup (T$_1$/T$_p$)}
& {\bf Time}\\ 
&
& {\bf (T$_{seq}$)}
& {\bf 1} & {\bf 8} & {\bf 16} & {\bf 24} & {\bf 32} &  {\bf (T$_{best}$)}\\
\hline\hline
\multicolumn{8}{l}{\bf Top-Down Approaches} \\
\multicolumn{1}{c}{\multirow{3}{*}{\bf YAP$_{TD_1}$}} 
& {\bf D$_{10}$} & 26,030 & 33,969 & {\bf 1.58} & 1.53 & 1.50 & 1.42 & 21,509\\
& {\bf D$_{30}$} & 26,523 & 34,213 & {\bf 1.60} & 1.54 & 1.50 & 1.42 & 21,424\\
& {\bf D$_{50}$} & 26,545 & 34,234 & {\bf 1.60} & 1.54 & 1.51 & 1.40 & 21,408\\
\hline
\multicolumn{1}{c}{\multirow{3}{*}{\bf YAP$_{TD_2}$}} 
& {\bf D$_{10}$} & 34,565 & 44,371 & 7.23 & 13.23 & 16.45 & {\bf 19.74} & 2,248\\
& {\bf D$_{30}$} & 34,284 & 44,191 & 7.12 & 13.09 & 16.52 & {\bf 19.77} & 2,235\\
& {\bf D$_{50}$} & 33,989 & 44,158 & 7.06 & 13.30 & 16.49 & {\bf 19.58} & 2,255\\
\hline
\multicolumn{8}{l}{\bf Bottom-Up Approaches} \\
\multicolumn{1}{c}{\multirow{3}{*}{\bf YAP$_{BU}$}} 
& {\bf D$_{10}$} & 20,799 & 28,909 & 6.47 & 12.21 & 16.48 & {\bf 20.32} & 1,423\\
& {\bf D$_{30}$} & 21,174 & 28,904 & 6.94 & 12.61 & 16.63 & {\bf 20.40} & 1,417\\
& {\bf D$_{50}$} & 21,166 & 28,857 & 6.44 & 12.31 & 16.44 & {\bf 20.52} & 1,406\\
\hline
\multicolumn{1}{c}{\multirow{3}{*}{\bf XSB$_{BU}$}} 
& {\bf D$_{10}$} & {\bf 60,983} & 74,108 & n.a. & n.a. & n.a. & n.a. &  60,983\\
& {\bf D$_{30}$} & {\bf 59,496} & 74,410 & n.a. & n.a. & n.a. & n.a. &  59,496\\
& {\bf D$_{50}$} & {\bf 59,700} & 74,628 & n.a. & n.a. & n.a. & n.a. &  59,700\\
\end{tabular}
\label{tab_lcs}
\end{table}

Table~\ref{tab_knapsack} and Table~\ref{tab_lcs} show the average of
10 runs results obtained, respectively, for the Knapsack and LCS
problems for both top-down and bottom-up approaches using the Yap and
XSB Prolog systems. The columns of both tables show the following
information. The first column describes the system and the dataset
used. The second column (T$_{seq}$) shows the sequential execution
time in milliseconds. For T$_{seq}$, the Prolog systems where compiled
without multithreaded support and ran without multithreaded code. The
next five columns show the execution time for one thread (T$_1$) and
the corresponding speedup for the execution with 8, 16, 24 and 32
threads (columns T$_1$/T$_p$). For each system/dataset configuration,
the results in bold highlight the column where the best execution time
was obtained and the last column (T$_{best}$) presents such result in
milliseconds.


Analyzing the general picture of both tables, one can observe that the
sequential time (T$_{seq}$) is always lower than the multithreaded
time (T$_{1}$). This is expected since the multithreaded version is
compiled and equipped with all the complex machinery required to
support concurrency in Yap, which includes not only all the new tabled
stuff but also all the base support for multithreaded in Yap.

When scaling the problem with multiple threads, the YAP$_{TD_2}$
top-down and YAP$_{BU}$ bottom-up approaches have the best results
with excellent speedups for 8, 16, 24 and 32 threads. In particular,
for 32 threads, they obtain speedups around 21 and 20, respectively,
for the Knapsack and LCS problems (T$_{1}$/T$_{best}$). If comparing
against the sequential version for 32 threads (not shown in the
tables), the speedups are around 15 and 16, respectively, for the
Knapsack and LCS problems (T$_{seq}$/T$_{best}$). The results for the
top-down YAP$_{TD_1}$ approach are not so interesting, regardless of
the fact that it can slightly scale for the Knapsack problem up to 16
threads.


Despite the similar average speedups for YAP$_{TD_2}$ and YAP$_{BU}$,
their execution times are quite different. Consider, for example, the
$D_{50}$ dataset of the Knapsack problem with 32 threads. While the
speedup $20.62$ of YAP$_{TD_2}$ corresponds to an execution time of
$1,233$ milliseconds, the speedup $21.76$ of YAP$_{BU}$ only
corresponds to $804$ milliseconds. Similarly, for the LCS problem, if
considering the $D_{50}$ dataset with 32 threads, while the speedup
$19.58$ of YAP$_{TD_2}$ corresponds to $2,255$ milliseconds, the
speedup $20.52$ of the YAP$_{BU}$ only corresponds to $1,406$
milliseconds.

The results also suggest that the execution times are not affected by
the values for the weights/profits generated. In general, the speedups
obtained for the different datasets ($D_{10}$, $D_{30}$ and $D_{50}$)
are always very close for the same number of threads.  Note that for
the bottom-up approaches this was expected since the complete matrix
of results has to be computed. For the top-down approaches, it can be
affected by the values for the weights/profits due to the depth in the
evaluation tree where solutions can be found. However, since we are
using randomized values in the datasets, we are aiming for the average
case.

Regarding the comparison with XSB's shared tables model, Yap's results
clearly outperform those of XSB. For the execution time with one
thread, XSB shows higher times than all Yap's approaches. For the
concurrent execution of the Knapsack problem, XSB shows no speedups
and for the concurrent execution of the LCS problem we have no results
available ($n.a.$) since we got \emph{segmentation fault} execution
errors. From our point of view, XSB's results are a consequence of the
\emph{usurpation operation}~\cite{Marques-08} that restricts the
potential of concurrency to non-mutually dependent
sub-computations. As the concurrent versions of the Knapsack and LCS
problems create mutual dependent sub-computations, which can be
executed in different threads, the XSB is actually unable to execute
them concurrently. In other works, even if we launch an arbitrary
large number of threads on those programs, the system would tend to
use only one thread at the end to evaluate most of the computations.


\section{Future Perspectives and Challenging Research Directions}

Currently, Yap provides the ground technology for both implicit and
explicit concurrent tabled evaluation, but separately. From the user's
point of view, tabling can be enabled through the use of single
directives of the form `\emph{:-~table p/n}', meaning that common
sub-computations for \emph{p/n} will be synchronized and shared
between workers at the engine level, i.e., at the level of the tables
where the results for such sub-computations are stored. Implicit
concurrent tabled evaluation can be triggered if using the OPTYap
design~\cite{Rocha-05a}, which exploits implicit or-parallelism using
shared memory processes. Explicit concurrent tabled evaluation can be
triggered if using the thread-based implementation~\cite{Areias-12a},
but the user still needs to explicitly implement the thread management
and scheduler policy for task distribution, which is orthogonal to the
focus of this work. Table~\ref{tab_concurrent_features} highlights the
key differences between the two concurrent tabling strategies in Yap's
current implementation.

\begin{table}[!ht]
\centering
\caption{Concurrent tabling supported features}
\begin{tabular}{c|cccc}
\multirow{2}{*}{\textbf{\emph{Strategy}}}
  & \textbf{\emph{Execution}} & \textbf{\emph{Memory}}    & \textbf{\emph{Synchronization}} & \textbf{\emph{Mode-Directed}} \\
  & \textbf{\emph{Model}}     & \textbf{\emph{Allocator}} & \textbf{\emph{Mechanisms}}      & \textbf{\emph{Tabling}} \\
\hline\hline
\textbf{\emph{Implicit}} & \emph{Processes/Threads} & \emph{Fixed-Size} & \emph{Lock-Based} & \emph{--} \\
\textbf{\emph{Explicit}} & \emph{Threads}           & \emph{Fixed-Size} & \emph{Lock-Free}  & \emph{NS/SS/PAS Designs} \\
\end{tabular}
\label{tab_concurrent_features}
\end{table}

The present work could thus be viewed as the basis to further
directions and further research in this area. So far, we have achieved
our initial goal. Even so, the system still has some restrictions that
may reduce its use elsewhere and its contribution to general Prolog
applications. We next discuss future perspectives and challenging
research directions:

\begin{description}
\item[Extend CS design to support lock-free data structures.] Due to
  the good performance results obtained with the lock-free proposals,
  an obvious research direction for further work is to extend the
  original CS design to use lock-free data structures instead of the
  lock-based data structures.
  
\item[Extend CS/FS/PAC designs to support mode-directed tabling.]  In
  the previous section, we observed the advantages of combining
  mode-directed tabling with the PAS design. However, in the PAS
  design, the answers to common tabled subgoal calls are only shared
  when the corresponding tables are completed. Since the CS/FS/PAC
  designs do not require the completion of tables to share answers,
  threads would be able to share and propagate answers sooner. The
  problem of combining mode-directed tabling with the CS/FS/PAC
  designs is on how to efficiently support concurrent delete
  operations on the trie structures and on how to efficiently handle
  the interface between consumer calls and the navigation in the trie
  of answers for the several running workers.
  
\item[Support concurrent delete operations on the trie structures.]
  As mention above, this is a key feature to allow for an efficient
  implementation of concurrent mode-directed tabling with the
  CS/FS/PAC designs. Moreover, this extension could also be applied to
  concurrent incremental tabling~\cite{Diptikalyan-PhD}, where
  specific subgoal calls and answers can be dynamically deleted during
  tabled evaluation.

\item[Concurrent linear tabling.] Since the evaluation of programs
  with a linear tabling engine is less complex than the evaluation
  with a suspension-based engine, it would be interesting to study how
  different linear tabled strategies~\cite{Areias-11,Areias-13} could
  run concurrently and take advantage of the different table space
  designs presented in this work.

\item[Implicit and explicit concurrent evaluation in a single
  framework.] This is our most challenging goal towards an efficient
  concurrent framework which integrates both implicit and explicit
  concurrent tabled evaluation in a single tabling engine. This is a
  very complex task since we need to combine the explicit control
  required to launch, assign and schedule tasks to workers, with the
  built-in mechanisms for handling tabling and/or implicit
  concurrency, which cannot be controlled by the user. In such a
  framework, a program begins as a single worker that executes
  sequentially until reaching an implicit or explicit concurrent
  construct. When reaching an explicit concurrent construct, the
  execution model launches a set of additional workers to exploit
  concurrently a set of independent sub-computations (which may
  include tabled and non-tabled predicates). From the workers point of
  view, each concurrent sub-computation computes its tables but, at
  the implementation level, the tables can be shared following the
  table space designs presented before for implicit concurrent tabled
  evaluation. Otherwise, if reaching an explicit concurrent construct,
  the execution model launches a set of additional workers to exploit
  in parallel a given sub-computation. Parallel execution is then
  handled implicitly by the execution model taking into account
  possible directive restrictions. For example, we may have directives
  to define the number of workers, the scheduling strategy to be used,
  load balancing policies, etc. By taking advantage of these explicit
  parallel constructs, a user can write parallel logic programs from
  scratch or parallelise existing sequential programs by incrementally
  pinpointing the sub-computations that can benefit from parallelism,
  using the available directives to test and fine tune the program in
  order to achieve the best performance. Such a framework could renew
  the glamour of Prolog systems, especially in the concurrent/parallel
  programming community. Combining the inherent implicit parallelism
  of Prolog with explicit high-level parallel constructs will clearly
  enhance the expressiveness and declarative style of tabling and
  simplify concurrent programming.
\end{description}





\bibliographystyle{acmtrans}
\bibliography{references}


\newpage
\appendix
\setcounter{theorem}{0}
\setcounter{corollary}{0}


\section{Proofs}
\label{appendix_proofs}

\begin{theorem}
If $NT \geq 1$ and $NC(P_i) \geq 1$ then $MU_{SS}(P_i) \leq
MU_{NS}(P_i) $ if and only if the formula $ [NC(P_i) - 1] * BA \leq
[NT - 1] * ST(P_i)$ holds.
\end{theorem}

\begin{proof}
Assuming that all tabled subgoal calls are completely evaluated, for
the NS design we have:

\begin{align*}
MU_{NS}(P_i) &= TE + BA + NT * ST(P_i) + NT * \sum\limits^{NC(P_i)}_{j=1}[SF + AT(P_{i.j})]
\end{align*}

And, for the SS design we have:

\begin{align*}
MU_{SS}(P_i) &= TE + ST(P_i) + \sum\limits^{NC(P_i)}_{j=1} [BA + NT * [SF + AT(P_{i.j})]]\\
            &= TE + ST(P_i) + NC(P_i) * BA + NT * \sum\limits^{NC(P_i)}_{j=1} [SF + AT(P_{i.j})]
\end{align*}

The value of $MU_{SS}(P_i) - MU_{NS}(P_i)$ is then given by:

\begin{align*}
MU_{SS}(P_i) - MU_{NS}(P_i) &= ST(P_i) + NC(P_i)*BA - BA - NT* ST(P_i) \\
                          &= [NC(P_i)-1] * BA - [NT - 1] * ST(P_i)
\end{align*}
        
Now, for the final part of the proof:

\begin{align*}
MU_{SS}(P_i) \leq MU_{NS}(P_i) & \Leftrightarrow MU_{SS}(P_i) - MU_{NS}(P_i) \leq 0 \\
& \Leftrightarrow [NC(P_i)-1] * BA - [NT - 1] * ST(P_i) \leq 0 \\
& \Leftrightarrow [NC(P_i) - 1] * BA \leq [NT-1] * ST(P_i)
\end{align*}
\end{proof}


\begin{theorem}
If $NT > 1$ and $NC(P_i) \geq 1$ then $MU_{FS}(P_i) < MU_{SS}(P_i)$.
\end{theorem}

\begin{proof}
Assuming that all tabled subgoal calls are completely evaluated, for
the FS design we have:

\begin{align*}
MU_{FS}(P_i) &= TE + ST(P_i) + \sum\limits^{NC(P_i)}_{j=1} [SE_{FS} + BA + NT * [SF_{FS} + BP] + AT(P_{i.j})] \\
            &= TE + ST(P_i) + \sum\limits^{NC(P_i)}_{j=1} [SF - SF_{FS} + BA + NT * [SF_{FS} + BP] + AT(P_{i.j})] \\
            &= TE + ST(P_i) + \sum\limits^{NC(P_i)}_{j=1} [SF + [NT-1]*SF_{FS} + BA + NT * BP + AT(P_{i.j})] \\
            &= TE + ST(P_i) + NC(P_i) * BA + \sum\limits^{NC(P_i)}_{j=1} [SF + [NT-1]*SF_{FS} + NT* BP + AT(P_{i.j})]
\end{align*}

And, for the SS design we have:

\begin{align*}
MU_{SS}(P_i) = TE + ST(P_i) + NC(P_i) * BA + NT * \sum\limits^{NC(P_i)}_{j=1} [SF + AT(P_{i.j})]
\end{align*}

The value of $MU_{FS}(P_i) - MU_{SS}(P_i)$ is then given by:

\begin{align*}
MU&_{FS}(P_i) - MU_{SS}(P_i) = \\
& = \sum\limits^{NC(P_i)}_{j=1} [SF + [NT-1]*SF_{FS} + NT * BP + AT(P_{i.j})] - NT * \sum\limits^{NC(P_i)}_{j=1} [SF + AT(P_{i.j})] \\
& = \sum\limits^{NC(P_i)}_{j=1} [NT * BP] + \sum\limits^{NC(P_i)}_{j=1} [SF + [NT-1]*SF_{FS} + AT(P_{i.j}) - NT * SF - NT * AT(P_{i.j})] \\
& = \sum\limits^{NC(P_i)}_{j=1} [NT * BP] + \sum\limits^{NC(P_i)}_{j=1} [[NT - 1] * [SF_{FS} - SF - AT(P_{i.j})] \\
& = \sum\limits^{NC(P_i)}_{j=1} [NT * BP] + \sum\limits^{NC(P_i)}_{j=1} [[NT-1]*[SF_{FS} - SF]] - \sum\limits^{NC(P_i)}_{j=1} [[NT-1]*AT(P_{i.j})]
\end{align*}

Now, for the final part of the proof:

\begin{align*}
MU&_{FS}(P_i) < MU_{SS}(P_i) \Leftrightarrow MU_{FS}(P_i) - MU_{SS}(P_i) < 0 \\
& \Leftrightarrow \sum\limits^{NC(P_i)}_{j=1} [NT * BP] + \sum\limits^{NC(P_i)}_{j=1} [[NT-1]*[SF_{FS} - SF]] - \sum\limits^{NC(P_i)}_{j=1} [[NT-1]*AT(P_{i.j})] < 0 \\
& \Leftrightarrow \sum\limits^{NC(P_i)}_{j=1} [NT * BP] + \sum\limits^{NC(P_i)}_{j=1} [[NT-1]*[SF_{FS} - SF]] < \sum\limits^{NC(P_i)}_{j=1} [[NT-1]*AT(P_{i.j})] \\
& \Leftrightarrow NT * BP + [NT-1]*[SF_{FS} - SF] < [NT-1]*AT(P_{i.j}) \\
& \Leftrightarrow \underbrace{[NT-1]*\underbrace{[SF_{FS} + BP - SF]}_{<0} + BP}_{<0} < [NT-1]* \underbrace{AT(P_{i.j})}_{>0}
\end{align*}
\end{proof}


\section{Benchmark Details}
\label{appendix_bechmark_details}

Table~\ref{tab_benchs} shows the characteristics of the five sets of
benchmark programs. The \emph{Large Joins} and \emph{WordNet} sets
were obtained from the OpenRuleBench project~\cite{Liang-09}; the
\emph{Model Checking} set includes three different specifications and
transition relation graphs usually used in model checking
applications; the \emph{Path Left} and \emph{Path Right} sets
implement two recursive definitions of the well-known $path/2$
predicate, that computes the transitive closure in a graph, using
several configurations of $edge/2$
facts. Figure~\ref{fig_edge_configurations} shows an example for each
configuration. We experimented the \emph{BTree} configuration with
depth 17, the \emph{Pyramid} and \emph{Cycle} configurations with
depth 2000 and the \emph{Grid} configuration with depth 35. All
benchmarks find all the solutions for the problem.

\begin{figure}[ht]
\centering
\includegraphics[width=11cm]{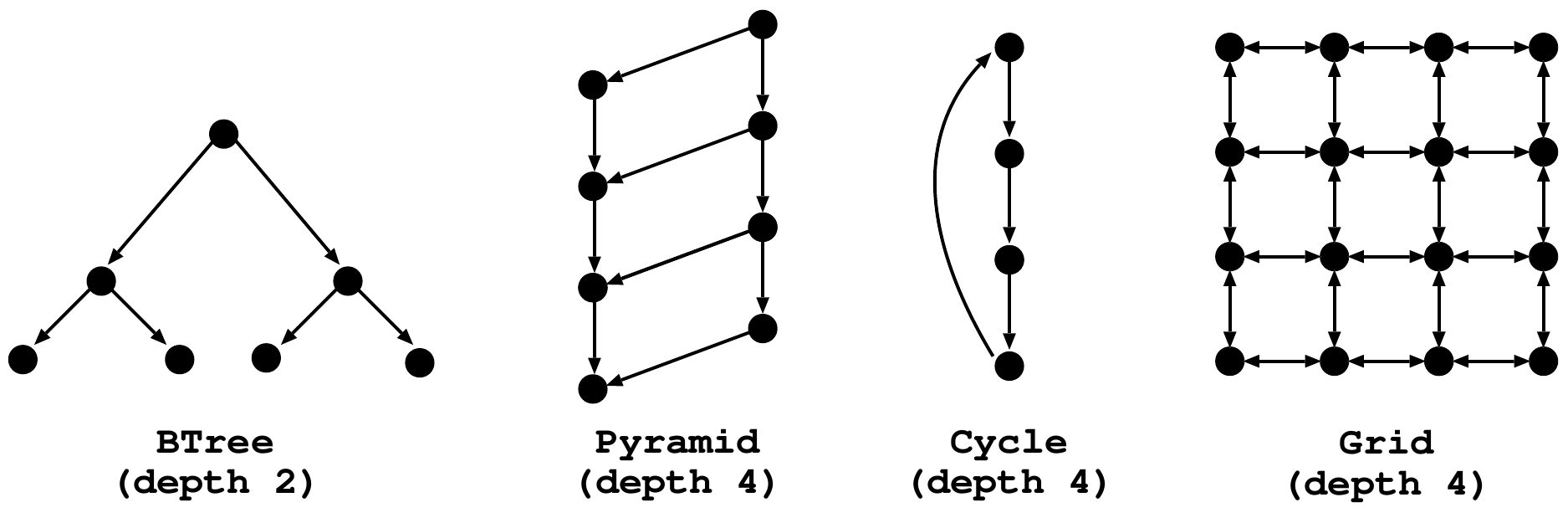}
\caption{Edge configurations for the path benchmarks}
\label{fig_edge_configurations}
\end{figure}

The columns in Table~\ref{tab_benchs} have the following meaning:

\begin{itemize}
\item {\bf calls:} is the number of different calls to tabled
  subgoals. It corresponds to the number of paths in the subgoal
  tries.
\item {\bf trie nodes:} is the total number of trie nodes allocated in
  the corresponding subgoal/answer trie structures.
\item {\bf trie depth:} is the minimum/average/maximum number of trie
  node levels required to represent a path in the corresponding
  subgoal/answer trie structures. Trie structures with smaller average
  values are more amenable to contention, i.e., to have a higher
  number of synchronization points.
\item {\bf unique:} is the number of different tabled answers
  found. It corresponds to the number of paths in the answer tries.
\item {\bf repeated:} is the number of redundant tabled answers
  found. 
\item {\bf NS:} is the average execution time, in seconds, of ten runs
  for 1 thread with the NS design.
\end{itemize}

\begin{sidewaystable}
\centering
  \caption{Characteristics of the benchmark programs}
\begin{tabular}{lrrccrrrcc}
\multirow{2}{*}{\bf Bench}
& \multicolumn{3}{c}{\bf Tabled Subgoals}
&
& \multicolumn{4}{c}{\bf Tabled Answers}
& \multicolumn{1}{c}{\bf Time (sec)} \\ \cline{2-4}\cline{6-9}
& \multicolumn{1}{c}{\bf calls}
& \multicolumn{1}{c}{\bf trie nodes}
& \multicolumn{1}{c}{\bf trie depth}
&
& \multicolumn{1}{c}{\bf unique}
& \multicolumn{1}{c}{\bf repeated}
& \multicolumn{1}{c}{\bf trie nodes}
& \multicolumn{1}{c}{\bf trie depth}
& \multicolumn{1}{c}{\bf NS} \\
\hline\hline
\multicolumn{10}{l}{\bf Large Joins} \\
\bf{Join2}    &       1 &       6 & 5/5/5 && 2,476,099 &         0 &
2,613,660 &    5/5/5 &  2.85 \\
\bf{Mondial}  &      35 &      42 & 3/4/4 &&     2,664 & 2,452,890 &
   14,334 &    6/7/7 &  0.84 \\
\hline
\multicolumn{10}{l}{\bf WordNet} \\
\bf{Clusters} & 117,659 & 235,319 & 2/2/2 &&   166,877 &   161,853 &
284,536 &    1/1/1 &   0.83 \\
\bf{Holo}     & 117,657 & 235,315 & 2/2/2 &&    74,838 &        54 &
192,495 &    1/1/1 &   0.75 \\
\bf{Hyper}    & 117,657 & 235,315 & 2/2/2 &&   698,472 &     8,658 &
816,129 &    1/1/1 &   1.42 \\
\bf{Hypo}     & 117,657 & 117,659 & 2/2/2 &&   698,472 &    20,341 &
816,129 &    1/1/1 &   1.53 \\
\bf{Mero}     & 117,657 & 117,659 & 2/2/2 &&    74,838 &        13 &
192,495 &    1/1/1 &   0.74 \\
\bf{Tropo}    & 117,657 & 235,315 & 2/2/2 &&       472 &         0 &
118,129 &    1/1/1 &   0.66 \\
\hline
\multicolumn{10}{l}{\bf Model Checking} \\
\bf{IProto}   &       1 &       6 & 5/5/5 &&   134,361 &   385,423 &
1,554,896 &  4/51/67 & 2.70 \\
\bf{Leader}   &       1 &       5 & 4/4/4 &&     1,728 &   574,786 &
   41,788 & 15/80/97 & 3.51  \\
\bf{Sieve}    &       1 &       7 & 6/6/6 &&       380 & 1,386,181 &
    8,624 & 21/53/58 & 18.50 \\
\hline
\multicolumn{10}{l}{\bf Path Left} \\
\bf{BTree}    &       1 &       3 & 2/2/2 && 1,966,082 &         0 &
2,031,618 &    2/2/2 &   1.53 \\
\bf{Cycle}    &       1 &       3 & 2/2/2 && 4,000,000 &     2,000 &
4,002,001 &    2/2/2 &   3.52 \\
\bf{Grid}     &       1 &       3 & 2/2/2 && 1,500,625 & 4,335,135 &
1,501,851 &    2/2/2 &  1.93 \\
\bf{Pyramid}  &       1 &       3 & 2/2/2 && 3,374,250 & 1,124,250 &
3,377,250 &    2/2/2 &   3.08 \\
\hline
\multicolumn{10}{l}{\bf Path Right} \\
\bf{BTree}    & 131,071 & 262,143 & 2/2/2 && 3,801,094 &         0 &
3,997,700 &    1/2/2 &  2.33 \\
\bf{Cycle}    &   2,001 &   4,003 & 2/2/2 && 8,000,000 &     4,000 &
8,004,001 &    1/2/2 &  3.55 \\
\bf{Grid}     &   1,226 &   2,453 & 2/2/2 && 3,001,250 & 8,670,270 &
3,003,701 &    1/2/2 &  2.32 \\
\bf{Pyramid}  &   3,000 &   6,001 & 2/2/2 && 6,745,501 & 2,247,001 &
6,751,500 &    1/2/2 &  3.17 \\
\end{tabular}
\label{tab_benchs}
\end{sidewaystable}

The \emph{Mondial} benchmark, from the \emph{Large Joins} set, and the
three \emph{Model Checking} benchmarks seem to be the benchmarks least
amenable to contention since they are the ones that find less unique
answers and that have the deepest trie structures. In this regard, the
\emph{Path Left} and \emph{Path Right} sets correspond to the opposite
case. They find a huge number of answers and have very shallow trie
structures. On the other hand, the \emph{WordNet} and \emph{Path
  Right} sets have the benchmarks with the largest number of different
subgoal calls, which can reduce the probability of contention because
answers can be found for different subgoal calls and therefore be
inserted with minimum overlap. On the opposite side are the
\emph{Join2} benchmark, from the \emph{Large Joins} set, and the
\emph{Path Left} benchmarks, which have only a single tabled subgoal
call.


\end{document}